\documentclass[11pt]{article}
\usepackage{graphicx}
\usepackage[utf8]{inputenc}
\usepackage{amsmath}
\usepackage{amssymb}
\usepackage{amsthm}
\usepackage[boxed,lined,linesnumbered]{algorithm2e}
\usepackage{enumerate}
\usepackage[dvipsnames]{xcolor}
\usepackage{nicefrac}
\usepackage{graphics}
\usepackage{graphicx}
\usepackage{wrapfig}
\usepackage{hyperref}
\usepackage[T1]{fontenc}
\usepackage{thmtools,thm-restate}
\usepackage{adjustbox}
\usepackage[framemethod=tikz]{mdframed}
\mdfsetup{frametitlealignment=\centering}
\usepackage{soul}
\usepackage[letterpaper,margin=1in]{geometry}
\usepackage{graphicx}
\usepackage{subcaption} 
\usepackage{tikz}
\usetikzlibrary{shapes.geometric}
\usetikzlibrary{arrows.meta}
\usetikzlibrary{positioning}

\newtheorem{theorem}{Theorem}[section]
\newtheorem{lemma}{Lemma}[section]

\newtheorem{claim}{Claim}[section]
\newtheorem{proposition}{Proposition}[section]

\newtheorem{observation}{Observation}[section]

\newcommand{\opt}{\ensuremath{\textsf{OPT}}}

\newcommand{\cost}{\ensuremath{\textsf{cost}}}

\newcommand{\poly}{\ensuremath{\textsf{poly}}}

\newcommand{\abs}[1]{\ensuremath{|#1|}}
\newcommand{\ball}{{\sf ball}}

\newcommand{\km}{\textsc{$k$-Median}\xspace}

\newcommand{\kzc}{\textsc{$(k,z)$-Clustering}\xspace}
\newcommand{\kmns}{\textsc{$k$-Means}\xspace}
\newcommand{\kc}{\textsc{$k$-Center}\xspace}
\newcommand{\sfm}{\textsc{Socially Fair $k$-Median}\xspace}

\newcommand{\fpt}{\textsf{FPT}\xspace}
\newcommand{\wone}{\textsf{W}$[1]$\xspace}

\newcommand{\mcis}{\textsc{Multi-Colored Independent Set}\xspace}

\newcommand{\rkc}{\textsc{Robust $(k,z)$-Clustering}\xspace}
\newcommand{\rkm}{\textsc{Robust $k$-Means}\xspace}
\newcommand{\an}{\textsf{ALG}_{\textnormal{n}}\xspace}
\newcommand{\af}{\textsf{ALG}_{\textnormal{f}}\xspace}
\newcommand{\on}{\textsf{OPT}_{\textnormal{n}}\xspace}
\newcommand{\of}{\textsf{OPT}_{\textnormal{f}}\xspace}
\newcommand{\bicf}{\textsf{BIC}_{\textnormal{f}}\xspace}
\newcommand{\bicn}{\textsf{BIC}_{\textnormal{n}}\xspace}
\newcommand{\bic}{\textsf{BIC}\xspace}
\newcommand{\interior}{{\sf interior}}

\newcommand{\defproblemout}[3]{
  \vspace{1mm}
\noindent\fbox{
  \begin{minipage}{0.96\textwidth}
  \begin{tabular*}{\textwidth}{@{\extracolsep{\fill}}lr} #1 \\ \end{tabular*}
  {\bf{Input:}} #2  \\
  {\bf{Output:}} #3
  \end{minipage}
  }
  \vspace{1mm}
}

\newcommand{\OO}{\mathcal{O}}

\newcommand{\bN}{\ensuremath{\mathbb{N}}\xspace}

\newcommand{\bR}{\ensuremath{\mathbb{R}}\xspace}
\newcommand{\bRd}{\ensuremath{\mathbb{R}^d}\xspace}

\newcommand{\bigO}[1]{\ensuremath{\mathcal{O}\left(#1\right)}}

\newcommand{\rd}{\mathbb{R}^d}

\newcommand{\vecw}{{\bf w}\xspace}
\newcommand{\vecd}{{\boldsymbol{\delta}}}

\newcommand{\vcw}{{\bf w}\xspace}

\newcommand{\cB}{\ensuremath{\mathcal{B}}\xspace}

\newcommand{\cE}{\ensuremath{\mathcal{E}}\xspace}

\newcommand{\cI}{\ensuremath{\mathcal{I}}\xspace}

\newcommand{\cM}{\ensuremath{\mathcal{M}}\xspace}

\newcommand{\cQ}{\ensuremath{\mathcal{Q}}\xspace}

\newcommand{\cW}{\ensuremath{\mathcal{W}}\xspace}

\newcommand{\kms}{\textsc{$k$-Means}}
\newcommand{\kmd}{\textsc{$k$-Median}}
\newcommand{\kcen}{\textsc{$k$-Center}}

\newtheoremstyle{special}
    {\topsep}
    {\topsep}
    {\itshape}    
    {}
    {\bfseries}
    {\,\((\star)\).}
    {.5em}
    {}

\title{Parameterized Approximation for Robust Clustering\\ in Discrete Geometric Spaces}

\hypersetup{pdfinfo={
  Title={Parameterized Approximation for Robust Clustering in Discrete Geometric Spaces},
  Author={Fateme Abbasi, Sandip Banerjee, Jarosław Byrka, Parinya Chalermsook, Ameet Gadekar, Kamyar Khodamoradi, D\'{a}niel Marx,\\
  Roohani Sharma, Joachim Spoerhase},
  Keywords={}
}}

\author{
Fateme Abbasi\footnote{University of Wroc\l{}aw, Poland (\texttt{fateme.abbasi@cs.uni.wroc.pl})} \and 
Sandip Banerjee\footnote{IDSIA, USI-SUPSI, Switzerland (\texttt{sandip.ndp@gmail.com})} \and 
Jaros\l{}aw Byrka\footnote{University of Wroc\l{}aw, Poland (\texttt{jby@cs.uni.wroc.pl})} 
\end{tabular} \endgraf
\begin{tabular}[t]{c}
Parinya Chalermsook\footnote{Aalto University, Finland (\texttt{parinya.chalermsook@aalto.fi})} \and
Ameet Gadekar\footnote{Bar-Ilan University, Israel (\texttt{ameet.gadekar@biu.ac.il})} \and
Kamyar Khodamoradi\footnote{University of Regina, Canada (\texttt{kamyar.khodamoradi@uregina.ca})} 
\end{tabular} \endgraf
\begin{tabular}[t]{c}
D\'{a}niel Marx\footnote{CISPA Helmholtz Center for Information Security, Saarbr\"{u}cken, Germany (\texttt{marx@cispa.de})} \and
Roohani Sharma\footnote{University of Bergen, Norway (\texttt{rsharma@uib.no})} \and
Joachim Spoerhase\footnote{University of Liverpool, United Kingdom (\texttt{joachim.spoerhase@liverpool.ac.uk})}
}

\date{}

\begin{document}

\maketitle

\begin{abstract}

We consider the well-studied \rkc problem, which generalizes the classic \kmd, \kmns, and \kcen~problems and arises in the domains of robust optimization [Anthony, Goyal, Gupta, Nagarajan, Math. Oper. Res. 2010] and in algorithmic fairness [Abbasi, Bhaskara, Venkatasubramanian, 2021 \& Ghadiri, Samadi, Vempala, 2022]. Given a constant $z\ge 1$, the input to \rkc is a set $P$ of $n$  points in a metric space $(M,\delta)$, a weight function $w: P \rightarrow {\mathbb R}_{\geq 0}$ and a positive integer $k$. Further, each point belongs to one (or more) of the $m$ many different groups $S_1,S_2,\ldots,S_m \subseteq P$. Our goal is to find a set $X$ of $k$ centers such that $\max_{i \in [m]} \sum_{p \in S_i} w(p) \delta(p,X)^z$ is minimized. 

 Complementing recent work on this problem, we give a comprehensive understanding of the 
 parameterized approximability of the problem in geometric spaces where the parameter is the number~$k$ of centers. We prove the following results: 
 \begin{enumerate}[(i)]
  \item For a universal constant $\eta_0 >0.0006$, we devise a $3^z(1-\eta_{0})$-factor \fpt approximation algorithm for \rkc in \emph{discrete} high-dimensional Euclidean spaces where the set of potential centers is finite. This shows that the lower bound of~$3^z$ for general metrics  [Goyal, Jaiswal, Inf. Proc. Letters, 2023] no longer holds when the metric has geometric structure.
 \item We show that \rkc in discrete Euclidean spaces is $(\sqrt{3/2}- o(1))$-hard to approximate for \fpt algorithms, even if we consider the special case \kc in logarithmic dimensions. This rules out a 
 $(1+\epsilon)$-approximation algorithm running in time $f(k,\epsilon)\poly(m,n)$ (also called efficient parameterized approximation scheme or EPAS), giving a striking contrast with the recent EPAS for the \emph{continuous} setting where centers can be placed anywhere in the space [Abbasi et al., FOCS'23].
 \item However, we obtain an EPAS for \rkc in discrete Euclidean spaces when the dimension is sublogarithmic (for the discrete problem, earlier work [Abbasi et al., FOCS'23] provides an EPAS only in dimension $o(\log\log n)$). Our EPAS works also for metrics of sub-logarithmic doubling dimension.  
\end{enumerate}    
\end{abstract}



\section{Introduction}\label{sec:introduction}

Clustering is a crucial method in the analysis of massive datasets and has widespread applications in operations research and machine learning. As a consequence, optimization problems related to clustering have received significant attention from the theoretical computer science community over the years. Within the framework of center-based clustering, \kcen , \kms, and \kmd ~\cite{HS85,JV01,BPRST17,AhmadianNSW17,kanungo2004local} are widely regarded as the most fundamental problems. 

A general notion that captures various classic clustering problems is referred to as \kzc in the literature, where $z\ge 1$ is a constant. In this type of problem, the input is a set $P$ of data points (clients), a set $F$ of centers (facilities), a metric $\delta$ on $P \cup F$, and a positive integer $k$. The goal is to find a set $C\subseteq F$ of $k$ facilities that minimizes the following cost function:
\begin{equation*}
    cost(C)=\sum_{p\in P}cost(p,C)
\end{equation*}
where $cost(p,C)=\delta(p,C)^z$ and $\delta(p,C)=\min_{c\in C} \delta(p,c)$. Note that $\kzc$ encapsulates the classical \kmd, and \kms~for $z=1$ and $z=2$, respectively .

Center-based clustering has cemented its place as an unsupervised learning method that has proven effective in modeling a variety of real-world problem. In most of the practical machine learning applications however, it is observed that the input data is rarely of high quality.

To tackle this challenge, we study a robust version of \kzc in this paper which can handle uncertainty in the input: Consider a situation where we do not have complete knowledge about the clients that will be served. In order to perform well despite this uncertainty, Anthony et al.~\cite{AGGN10} defined a concept of robustness for the \kmd~problem, in which each possible scenario is represented by a group of clients and the goal is to find a solution that performs best possible even in the worst scenario. In this paper, we address the following robust version of the \kzc problem (called \rkc): \\

\defproblemout{\rkc}{Instance $(P,F,\delta)$ with $\delta$ being a metric on $P \cup F$, positive integer $k$, a weight function $w\colon P \rightarrow {\mathbb R}_{+}$, and $m$ groups $S_1,\ldots, S_m$ such that $S_i \subseteq P, P= \cup_{i \in [m]} S_i$.}{A $k$-element subset $X \subseteq F$ that minimizes $\max_{i \in [m]} \sum_{p \in S_i} w(p) \delta(p,X)^z$.}\\

Let $n = |P|$. We remark that, in addition to generalizing \kmd~and \kmns, the \rkc problem encapsulates \kcen, when each group contains a distinct singleton.
A similar objective has been studied in the context of {\it fairness}, in which we aim to create a solution that will be appropriate for each of the specified groups of people. This problem is known in the literature as \sfm, recently introduced independently by Abbasi et al.~\cite{ABV21} and Ghadiri et al.~\cite{ghadiri2021socially}.  Notice that Abbasi et al.~\cite{ABV21} introduce fair clustering with client weights being inversely proportional to the group size as a normalization. On the other hand, Anthony et al.~\cite{AGGN10} introduce robust clustering with unweighted clients. Since our definition allows arbitrary client weights, we capture both of these settings.

While \kms, \kmd, and \kc admit constant-factor approximations, it is not very surprising that \rkc is harder due to its generality: Makarychev and Vakilian~\cite{pmlr-v134-makarychev21a} design
a polynomial-time $\bigO{\log m/ \log \log m}$-approximation algorithm, which is tight under a plausible complexity assumption~\cite{bhattacharya2014new}\footnote{Note that they proved this factor for \textsc{Robust }\kmd, and the hardness result holds even in the line metric, unless $\textsf{NP} \subseteq \cap_{\delta >0} \textsf{DTIME}(2^{n^\delta})$.}. 
As this precludes the existence of efficient constant-factor approximation algorithms, recent works have focused on designing constant factor \emph{parameterized} (\fpt) approximation algorithms\footnote{ Throughout the paper, parameterization refers to the natural parameter $k$.}. Along these lines, an \fpt time $(3^z+\epsilon)$-approximation algorithm has been proposed and shown to be tight under the Gap Exponential-Time Hypothesis (Gap-ETH)~\cite{goyal2021tight}. When allowing a parameterization on the number of groups $m$ (instead of $k$), Ghadiri et al. designed a $(5 + 2 \sqrt{6} + \epsilon)^z$-approximation algorithm in $n^{\bigO{m^2}}$ time~\cite{ghadiri2022constant}.   

Motivated by the tight lower bounds for general discrete metrics, we focus on \emph{geometric}  spaces.  
Geometric spaces have a particular importance in real-world applications because data can often be represented via a (potentially large) collection of numerical attributes, that is, by vectors in a (possibly high-dimensional) geometric space. For example, in the bag-of-words model a document is represented by a vector where each coordinate specifies the frequency of a given word in that document. Such representations naturally lead to very high-dimensional data. A setting of particular interest is the  high-dimensional \emph{Euclidean space} where the metric is simply the Euclidean metric $\delta(x,y) = ||x-y||_2$.

The study of clustering problems in high-dimensional Euclidean space is an important line of research that has received significant attention in the algorithms community. It may seem intuitive to believe that it should generally (for almost any problem) be possible to algorithmically leverage the geometric structure to separate high-dimensional Euclidean from general metrics. For clustering, however, this turns out to be either false or highly non-trivial in many cases. For example, it is a long-standing open question~\cite{feder-greene88:approx-clustering} whether \kc admits a polynomial time $(2-\epsilon)$-approximation algorithm even in $\bR^2$, improving the tight bound of $2$ in general metrics. Interestingly enough, for the more general Euclidean \textsc{$k$-Supplier} problem, Nagarajan et al.~\cite{nagarajan-etal20:euclidean-k-supplier} obtain an improvement over the tight bound of $3$ in general metrics. The improved bounds for Euclidean \km and \kmns by Ahmadian et al.~\cite{AhmadianNSW17},  Grandoni et al.~\cite{GrandoniORSV22}, and recently by Cohen-Addad et al.~\cite{cohen-addad-etal22:improved-euclidean-k-median-k-means} were breakthroughs.
Concerning the more general \rkc, the tight inapproximability bound of $\Omega(\log m/ \log \log m)$ in general metric continues to hold even in the line metric~\cite{bhattacharya2014new}.

Similarly, the regime of \fpt approximation algorithms for Euclidean clustering problems has received significant attention. Classic works design an Efficient Parameterized Approximation Scheme (EPAS), that is, a $(1+\epsilon)$-approximation in $f(k,\epsilon) \poly(n)$ time, for \kc~\cite{badoiu-etal:approximate-clustering-coresets} as well as for \km and \kmns~\cite{kumar2010linear}. Recent research focuses on the design of so-called coresets~\cite{sohler-woodruff18:coresets-k-median,cohen-addad-etal21:coreset-framework} whose existence implies an EPAS if their size only depends on $k$ and the error parameter $\epsilon$.

In the real space $\bRd$, it is important to distinguish between the \emph{discrete} and the \emph{continuous} settings. In the discrete setting, both the point set $P$ and the candidate center set $F$ are finite subsets of ${\mathbb R}^d$ while in the continuous setting, centers can be chosen anywhere in the metric space, that is, $F = {\mathbb R}^d$. A separate line of research has studied the contrast between continuous and discrete versions. For example, while discrete clustering variants are clearly polynomial-time solvable for constant $k$ by trivial enumeration, the continuous versions of \kc and \kmd~are known to be NP-hard even for $k=2$~\cite{drineas-etal04:np-hardness-2-means-highdim} in high-dimensional Euclidean space. Also in terms of polynomial-time approximability, stronger lower bounds were shown by Cohen-Addad et al.~\cite{cohen2021approximability} for the continuous versions. Indeed, there have been systematic research efforts in understanding these geometric clustering problems~\cite{cohen2022johnson,cohen2019inapproximability,cohen2021approximability}.
A recent result~\cite{abbasi-etal23:epas-norm-clustering} implies an EPAS for \rkc in continuous Euclidean spaces (of any dimension), as well as in discrete Euclidean spaces in ``relatively low'' dimension, that is, dimension $o(\log \log n)$. 

The main goal of this paper is to develop comprehensive understanding for \rkc in high-dimensional discrete Euclidean spaces, in particular, when the dimension is at least $\Omega(\log \log n)$.

\bigskip
\noindent\textbf{Remark:} A preliminary version of this paper appeared at ICALP 2024 \cite{Icalpabbasi}. This manuscript represents the full version

\subsection{Our contributions}

First, motivated by a factor-$(3^z-o(1))$ hardness of \fpt approximation for \rkc in general metrics~\cite{goyal2021tight}, a natural question is whether the structures of Euclidean spaces can be leveraged to obtain better results in high dimensions. 
While it is intuitive to believe that such an improvement should generally (for almost any problem) be possible in geometric spaces, we note that this is sometimes not the case: The polynomial time inapproximability of $\rkc$ remains $\Omega(\log m/ \log \log m)$ even in the line metric~\cite{bhattacharya2014new}.

Our first result gives an affirmative answer to this question. 

\begin{restatable}[High-Dimensional Euclidean Space]{theorem}{constapproxsfcz}
\label{thm:disfkz}
There exists a universal constant  $\eta_0 >0.0006$ such that for any constant positive integer $z$, there is a factor $3^z(1-\eta_0)$ \fpt approximation algorithm for \rkc in discrete Euclidean space ${\mathbb R}^d$ that runs in time
$2^{\bigO{k \log k}} \poly(m,n,d)$.
\end{restatable}

We remark that, first, our running time has only a polynomial dependency on $d$. Secondly, the key take-home message for Theorem~\ref{thm:disfkz} is not about a concrete approximation factor, but rather a ``proof of concept'' that the factor of $3^z$ can be improved.
Conceptually, this result shows that geometric spaces are indeed easier for \rkc than general metric spaces in the \fpt world, in contrast to the polynomial-time world, where they seem to be equally hard~\cite{bhattacharya2014new}.
The proof of this theorem relies on a new geometric insight that leverages the properties of Euclidean spaces (that do not hold in general metric spaces).
The analysis of our algorithms ``reduces'' the global analysis of approximation factor to a ``local'' geometric instance, in which it suffices to merely analyze the behavior of three points in the Euclidean spaces.

Next, we focus on obtaining a complete characterization of the existence of EPAS in discrete Euclidean spaces. 
Recall that an EPAS exists in continuous Euclidean spaces of any dimensions and in discrete Euclidean spaces of dimension $o(\log \log n)$~\cite{abbasi-etal23:epas-norm-clustering}, so to complete the landscape, we need to understand the discrete Euclidean spaces of dimension $\Omega(\log \log n)$.  

In the next theorem, we prove that even the special case of \kc does not admit an EPAS. This hardness holds for any $\ell_q$ metric and even in dimension $O(k \log n)$. More formally, we prove the following theorem.

\begin{restatable} [Hardness in Discrete Euclidean Space]{theorem}{hardnesseuc}\label{thm: LB-intro}
    For any constant positive integer $q$ and any positive constant $\eta>0$, there exists a function $d(k,n)=O(k\log n)$ such that there is no factor-$(3/2-\eta)^{1/q}$ $\fpt$ approximation algorithm for the discrete \kc problem in $\bR^{d(k,n)}$  under the $\ell_q$ metric unless $\wone=\fpt$. Moreover, for the $\ell_2$ metric this hardness holds even for some dimension $O(\log n)$, that is, independently of $k$.
\end{restatable}

Our result therefore highlights the interesting contrast between the discrete and continuous settings in high-dimensional Euclidean spaces, which has been systematically studied in recent years~\cite{cohen2022johnson,cohen2019inapproximability,cohen2021approximability}.
As mentioned, the continuous setting admits an EPAS~\cite{abbasi-etal23:epas-norm-clustering}, so  our hardness result  implies that the discrete setting is harder than the continuous counterpart. 
This is contrast to the results Cohen-Addad et al.~\cite{cohen2021approximability} mentioned earlier showing that continuous variants of \km and \kmns in geometric spaces  are apparently harder to approximate (in polynomial time) than their discrete part as well as the different complexity status of continuous and discrete clustering in high-dimensional spaces even for $k=2$~\cite{drineas-etal04:np-hardness-2-means-highdim}. This shows a rather mysterious behavior of clustering problems in geometric spaces.

Our next theorem completes the \fpt-approximability landscape by designing an EPAS for the problem in doubling metrics of dimension $d = o_k(\log n)$\footnote{We use notation $o_k(\cdot)$ to hide multiplicative factors depending only on $k$.}. We remark that the doubling dimension of the $d$-dimensional discrete Euclidean metric is $\Theta(d)$, that is, we obtain an EPAS for discrete Euclidean $o_k(\log n)$-dimensional spaces in particular.

\begin{restatable} [EPAS for Doubling Metric of Sub-Logarithmic Dimension]{theorem}{epasdoubling}\label{thm:pas_kd}

There is an algorithm that computes $(1+\epsilon)$-approximate solution, for every $\epsilon>0$,  for \rkc in the metric of doubling dimension $d$ in time $f(k,d,\epsilon,z) \poly(m,n)$, where $f(k,d,\epsilon,z) = \left(\left(\frac{2^z}{\epsilon}\right)^d k \log k\right)^{\bigO{k}}$.

\end{restatable}

Note that the above theorem yields an EPAS for \rkc when $d=o_k(\log n)$.
Together with Theorem~\ref{thm: LB-intro}, this theorem gives (almost) a dichotomy result for the existence of EPAS: An EPAS exists for \rkc in $o_k(\log n)$ dimension, while obtaining an EPAS is W[1]-hard in $\Omega_k(\log n)$ dimension. This leads to an almost complete understanding on the existence of EPAS in continuous and discrete Euclidean spaces.

\section{Overview of Techniques} 
\label{sec:overview}

\paragraph*{Improved \fpt Approximation in High-Dimensional Discrete Euclidean Space.}
Our algorithm underlying Theorem~\ref{thm:disfkz} is a slight modification of the factor-$(3^z+\epsilon)$ \fpt approximation algorithm for general metrics by Goyal and Jaiswal~\cite{goyal2021tight}. Our main technical contribution lies in the improved analysis. A key component of the analysis by Goyal and Jaiswal is a simple projection property of metric spaces (see Lemma~\ref{lem:projection} below). We argue that under minor additional assumptions, this property can be strengthened in Euclidean space. The resulting \emph{assigment lemma} (see Lemma~\ref{lem:assignment}) is at the heart of our analysis and its proof relies on several new ideas and technically involved ingredients.

We briefly review the algorithm by Goyal and Jaiswal~\cite{goyal2021tight}. Their algorithm consists of two main steps. First, they compute a $(\kappa,\lambda)$-bicriteria solution $B\subseteq F$, that is, the cost of~$B$ is bounded by $\kappa\opt$ and the cardinality of $B$ is bounded by~$\lambda k$. Specifically, they obtain guarantees $\kappa=1+\epsilon$ and $\lambda=\bigO{\log^2n/\epsilon^2}$ for sufficiently small $\epsilon>0$. In the second step, they extract a feasible solution from the (infeasible) bi-criteria solution $B$ by enumerating all $k$-subsets of~$B$ and outputting the one of minimum cost.

Their analysis is based on proving the existence of a $k$-subset of $B$ whose cost is at most ${(3^{z-1}(\kappa+2))\opt}$, which can be bounded by $(3^z+\epsilon)\opt$ assuming $z$ being constant. Since the algorithm enumerates all $k$-subsets, this provides an upper bound on the cost of the algorithm. The key component of their existential argument is the following simple property of metric spaces, which we call \emph{projection lemma}. It is convenient to think of $O$ as an optimal solution and $B$ as a bicriteria solution with $|B|>|O|$ but the lemma holds for any sets $B,O$.
\begin{restatable}[Projection Lemma]{lemma}{projection}
\label{lem:projection}
Let $(Y,\delta)$ be a metric space, and $B \subseteq Y$. Then for any set $O \subseteq Y$, there exists an assignment $\sigma\colon O \rightarrow B$ such that, for all $o \in O$ and $y \in Y$, we have 
\begin{equation}\label{eq:projection-lemma}
    \delta(y, \sigma(o)) \leq 2\delta(y, o)+ \delta(y,B)\,.
\end{equation} 
\end{restatable}

Intuitively, their lemma allows them to ``project'' the optimal solution $O$ onto a $k$-subset $\sigma(O) \subseteq B$ of the bicriteria solution so that for any client $y\in Y$, the distance $\delta(y,\sigma(O))$ can be charged to $\delta(y,O)$ and $\delta(y,B)$. If fact, the  number $3$ in the approximation factor $3^z+\epsilon$ corresponds to the sum $(2+1)$ of the coefficients in front of $\delta(y,o)$ and $\delta(y,B)$.

In this paper, we study the setting where $Y$ is a discrete Euclidean metric $(P,F,\delta)$, that is, where $P,F$ are finite subsets of $\bRd$ and $\delta$ is the Euclidean distance. A natural attempt to improve the approximation factor in the Euclidean setting is to reduce the coefficients in front of the terms $\delta(y,o)$ and $\delta(y,B)$ in the projection lemma. Unfortunately, this straightforward approach fails: The projection lemma is tight even on the line metric; see Figure \ref{pic:tight}.

\begin{figure}[!ht]
\input{pictures}
    \centering
    \tikzVI
    \vspace{4ex}
    \caption{\footnotesize This example shows that the projection lemma is tight even for the $1$-dimensional Euclidean space.
    Let $o=0$ be the optimum facility located at the origin and serving client $p=1/2$. Let $b'=1$ be the facility in $B$ that serves $p$ and let $b=\sigma(o)=-1$ be the facility in $B$ nearest to $o$. We have $\opt=1/2$, which also equals the cost of $B$. However $\delta(p,\sigma(o))=3/2=2\times \delta(p,o) + 1\times \delta(p,b')$. Combining multiple such examples in orthogonal directions and sharing facility $b$ shows that the approximation ratio of the algorithm of Goyal and Jaiswal~\cite{goyal2021tight} approaches $3$ in the discrete Euclidean space.}
  \label{pic:tight} 
\end{figure}

It turns out that slightly enlarging the projection space is already sufficient to bypass this obstacle. More specifically, we project onto the \emph{midpoint closure}
\begin{equation}\label{eq:midpt-closure}
  {\sf cl}(B) = B \cup \left\{\pi_F\left(\frac{b+b'}{2}\right)\colon b, b' \in B\right\}\, ,
\end{equation}
 of the bicriteria solution where $\pi_F(p)$ represents the closest facility in $F$ to point~$p$. This step exploits that the metric space is embedded into $\bRd$ (so that the midpoints exist).

While on the algorithmic side a slight modification of the original algorithm is sufficient for the improvement, the analysis requires several new ideas and technically involved ingredients. To prove a strengthened version of the projection lemma (called  assignment lemma) we set up a factor-revealing geometric optimization problem in the plane; see~(\ref{eq:displacement-ratio}) in Definition~\ref{def:gamma_beta} below. We call the optimum objective $\gamma_\beta$ of this problem \emph{displacement ratio}. Roughly speaking, this ratio corresponds to the maximum ratio between the left-hand and the right-hand side of~(\ref{eq:projection-lemma}) in Lemma~\ref{lem:projection}. However, we project to $\textsf{cl}(B)$ rather than $B$ and impose some additional minor restrictions. By a careful and technically involved analysis of this optimization problem we can upper bound the displacement ratio in the Euclidean setting by $1-\epsilon_0$ for some universal constant $\epsilon_0>0$ as long as two obstructions are avoided. The first obstruction occurs in any configuration similar to the one in Figure~\ref{pic:tight} above where the bi-criteria solution contains two facilities $b,b'$ so that $o$ is near to the mid-point of $b$ and $b'$. However, in such a configuration facility $o$ certifies that $b''=\pi_F((b+b')/2)$ must be close to $o$ allowing us to assign~$o$ to~$b''$ contained in the mid-point closure. The second obstruction arises if $p$ is \emph{$\beta$-near}, that is, within a small distance $\beta$ from $o$ (but there is no facility in $B$ such as $b'$ as in the first obstruction). For $\beta$ approaching $0$, the displacement ratio of $\beta$-near points can approach $1$ even if when projecting to the mid-point closure of $B$. To account for $\beta$-near points, we therefore cannot resort to the assignment lemma. However, the overall contribution of $\beta$-near  points to the cost of the projected solution can be shown to be very small. More details of the algorithm and its analysis are provided in Section~\ref{sec:FPT-kzc}. 

\paragraph*{Hardness of Discrete \kc}

Our proof constructs an instance of the discrete \kc  from an instance of \mcis problem, which is known to be $W[1]$-hard. In \mcis, we are given a $k$-partite graph $G$ with a $k$-partition of the vertices $V_1,\dots,V_k$, and  the goal is to determine if there is an independent set that contains precisely one node from each set $V_i$, $i\in[k]$.  
The gadget in our construction is a set of nearly equidistant binary code words. Such code words with relative Hamming distance roughly $\nicefrac{1}{2}$ and logarithmic length are known to exist (see Ta-Shma~\cite{ta-shma17:explicit-balanced-codes}). The high level idea is as follows. We associate each vertex of $G$ with a unique code word of suitable length $t$.
Then, we generate a data point in $P$ for each vertex and edge of $G$ by using code word(s) associated with the corresponding vertices. The construction guarantees the following crucial properties: (i) The Hamming distance between the data points of vertices is roughly $t$. (ii) The Hamming distance between 
a data point of vertex $v \in V_i$ and a data point of an edge $e$ is  roughly $t$ if $e$ is incident on $V_i \setminus \{v\}$ and is roughly $3t/2$ otherwise.
(iii) The Hamming distance between the data points of edges is at least (close to) $3t/2$. Thus, the construction forces us to pick data points of vertices as centers in our solution and guarantees that the optimum cost of the \kcen~instance is roughly $t$ if and only if there is an independent set in $G$. As a result, approximating the cost of the \kc instance better than a (roughly) $(3 / 2)^{1 / q}$ factor would imply $W[1] = \fpt$. That is because the cost of a \kc instance is the maximum $\ell_q$ distance between a data point and its closest selected center, and hence, approximating this cost better than the mentioned factor allows us to distinguish between \textsc{Yes} and \textsc{No} cases of an arbitrary instance of \mcis.

\paragraph*{Approximation Scheme for Metrics of Sub-Logarithmic Doubling Dimension.} 

Our algorithm comprises two main components, both based on standard techniques from the literature: instance compression and decomposition of the doubling metric into smaller balls. However, it becomes evident that a natural construction based on these standard techniques for \rkc  faces serious information-theoretic limitations, as explained below.
One natural idea for compressing a \rkc instance is to reduce the number of groups, as each group can be further compressed using a \kzc coreset (such coresets exist~\cite{cohen-addad-etal21:coreset-framework}). This reduction yields a significantly smaller instance. If we could reduce the number of groups to $m' \ll m$ while approximately preserving the cost for every solution, we could obtain an EPAS as follows.
First, apply a $\kzc$ coreset to every group of the compressed instance to obtain another \rkc instance with $m'$ groups, each containing $g(k,\epsilon)$ points, where $g$ is some function that represents the size of $\kzc$ coreset. It is essential to note that this compression is acceptable for obtaining an EPAS since the coreset of a group approximately preserves the \kzc cost of the group.
Next, enumerate all $k$-partitions of the points within each group to find potential solutions. Finally, return the solution that has the minimum \rkc cost. Unfortunately, because \rkc captures \kc (and consequently faces a coreset lower bound of $2^{\Omega(d)}$ in Euclidean space of dimension $d$~\cite{braverman2019coresets}), the number of new groups must satisfy $m' \ge 2^{\Omega(d)}$. Consequently, the running time of this algorithm is $k^{2^{\Omega(d)}}\poly(n,m)$, which is doubly exponential in $d$. It is worth noting that this algorithm matches the running time of~\cite{abbasi-etal23:epas-norm-clustering} and does not yield an EPAS for sub-logarithmic dimension.

Furthermore, if we explore an alternative approach and utilize the coreset of \kc, it is not immediately clear how to extend the coreset of \kc to reduce the number of groups in an instance of \rkc. This is because, firstly, we would require a mapping between the old groups and the new groups, and secondly, this mapping should ideally approximately preserve the \rkc cost for every solution.

Another potential method for compressing the instance involves reducing the number of points in set $P$, rather than altering the groups, with the hope of designing an EPAS that can exploit the smaller $P$ (without concern for the number of groups). However, for this approach to succeed, it is essential to establish a bijection between the old and new groups. Yet, it remains uncertain whether such a bijection exists. In typical coreset constructions, each point in the coreset $P'$ of $P$ has a weight that is the sum of the weights of the points in its local neighborhood in $P$ which it is supposed to represent in $P'$. However, these points in $P$ could potentially belong to different groups, making it challenging to establish the mapping between groups.

The core idea of our approach is to work with an alternative and more general definition of groups that permits a point to participate in different groups with varying weights. In this revised definition, instead of viewing groups as subsets of points, we treat each group as a weight function that assigns non-negative real values to points. This flexibility allows different weights to be assigned to the same point by different groups, which can, in fact, be of practical interest.
Utilizing this new definition, we can devise an approach for compressing the points such that each point in the compressed instance can have a weight for group $g$ that represents the sum of the weights of nearby points in $g$ that were filtered out during compression. Essentially, this enables us to approximately preserve the group costs. With this approach and additional technical work that leverages the standard ball decomposition technique for doubling metrics, we derive a coreset for \rkc that can be employed to construct an EPAS for doubling metrics with sub-logarithmic dimension.

\section{High-Dimensional Discrete Euclidean Space}
\label{sec:discrete}

\subsection{\fpt Approximation Algorithm for \rkc } \label{sec:FPT-kzc}

In this section, we exploit non-trivial properties of the Euclidean metric to prove the following result that breaches the barrier of $3^z$-approximation for \rkc in general metrics.

\constapproxsfcz*

Recall from Section \ref{sec:overview} that our approach begins with computing a $(\kappa,\lambda)$-bicriteria solution~$B$  to the \rkc instance employing the algorithm proposed by Goyal-Jaiswal \cite{goyal2021tight}. As we argued, it is sufficient to prove the existence of a $k$-subset of~$B$ whose cost is within a constant factor of optimal. The result by Goyal and Jaiswal~\cite{goyal2021tight} is based on the following simple projection lemma for general metrics whose proof we state here for the sake of later reference.

\projection*
\begin{proof}
    For each $o \in O$, define $\sigma(o)$ as $\pi_B(o)$, the point in $B$ closest in distance to $o$. Notice that for any $o\in O$, $y\in Y$, we have $\delta(y, \sigma(o)) \leq \delta(y,o) + \delta(o, \sigma(o))$ by triangle inequality. The lemma follows by combining this with $\delta(o, \sigma(o)) = \delta(o,B) \leq \delta(o, \pi_B(y)) \leq \delta(y, o)+\delta(y,B)$.
\end{proof}

This lemma itself is tight even in $1$-dimensional Euclidean space (as we showed in Figure~\ref{pic:tight}). In order to get around this issue, we make use of the property of our geometric space. Given the instance ($P, F, \delta)$ embedded into the Euclidean space and the bicriteria solution~$B$, we project to the mid-point closure ${\sf cl}(B)$ as defined in~(\ref{eq:midpt-closure}).

Notice that $|{\sf cl}(B)| = \bigO{|B|^2}$. Let $O$ be the optimal solution. For $\beta >0$ we say that client $p \in P$ is \emph{$\beta$-far} (from $O$ w.r.t. $B$) if $\delta(p,O) \geq \beta \cdot \delta(p,B)$, and we say that client $p$ is \emph{$\beta$-near} otherwise.
 The key of our analysis is the following strengthening of the projection lemma for Euclidean space, which we call assignment lemma.
 
\begin{restatable}[Assignment Lemma]{lemma}{assignmentlem}
\label{lem:assignment}
 Let $\beta_0 = 0.05$ and let $B \subseteq {\mathbb R}^d$. Then, for any $O \subseteq {\mathbb R}^d$, there exists an assignment $\sigma\colon O \rightarrow {\sf cl}(B)$ such that, for all $\beta_0$-far points $p \in {\mathbb R}^d$, we have $\delta(p, \sigma(O)) \leq (1-\epsilon_0)(2 \delta(p, O) + \delta(p, B))$ where $\epsilon_0 > 0.002$.
\end{restatable}
    
To prove Lemma \ref{lem:assignment}, we start with defining the assignment function $\sigma$. 
Take any facility $o \in O$ and let $b = \pi_B(o)$. We assume w.l.o.g.\ that the instance is rotated so that $p,b$ and $o$ lie in the plane spanned by the first two coordinates. For the sake of easier notation, we identify $p,b,o$ by points in $\bR^2$. Further, by translation and scaling, we assume that $o$ coincides with the origin and that $b=(-1,0)$. Let $q=(0,1)$ be the mirror image of~$b$. Let $\alpha$ be a parameter to be fixed (we later set it to $0.6$). We define $\sigma(o)$ based on the position of~$o$ relative to an \(\alpha\)-ball. Specifically, $\sigma(o) = b$ if the $\alpha$-ball centered at a point $q$ contains no facility from $B$; otherwise, $\sigma(o)$ is the projection $\pi_{\text{cl}(B)}(o)$ of $o$ onto the mid-point closure of~$B$.

Our goal is to analyze the displacement of a client $p$  under the assignment rule $\sigma$. Recall from the proof of Lemma \ref{lem:projection} that if $\sigma(o)$ is simply the projection onto $B$, then a client $p$, when served by facilities $o$ and $b'$ in sets $O$ and $B$ respectively, incurs a cost of at most $2||p - o|| + ||p - b'||$. We wish to show that the assignment cost in our algorithm is strictly smaller than this upper bound (under certain assumptions). We prove this by bounding the ratio of these two quantities.

\begin{restatable}[Displacement Ratio]{definition}{gammabeta}
\label{def:gamma_beta}
For a given small constant $\beta > 0$, let the \emph{displacement ratio} be defined as
 \begin{equation}\label{eq:displacement-ratio}
 \gamma_\beta = \max_{\substack{p\in \bR^d \backslash \ball(o, \beta),\\ b'\in\bR^d\setminus\ball(o,1) }} \left \{ \frac{||p-\sigma(o)||}{2||p - o|| + ||p - b'||} \right\}\,.
 \end{equation}
\end{restatable}
Let $S$ be the plane spanned by $b$, $p$, and $o$. After the appropriate rotations and translations we mentioned earlier, $S$ would coincide with the $x$-$y$ plane. In what follows, we also restrict $b'$ to lie in  plane $S$ as well. This follows due to Claim~\ref{lem:restrict-to-plane} mentioned below. For cleaner analysis, we defer the proof of this claim to Section~\ref{ss: lem plane}.
\begin{claim}\label{lem:restrict-to-plane}
The maximum displacement ratio $\gamma_\beta$ is achieved by some $b'$ that lies in plane $S$ containing points $b$, $o$, and $p$.
\end{claim}

To show the lemma, we demonstrate that $\gamma_{\beta}$ can be upper-bounded by $1-f(\alpha , \beta)$ for some $f(\alpha , \beta)>0$, where $f(\cdot)$ is a function dependent on $\alpha$, $\beta$ and the geometry of $O$ and~$B$.

\begin{claim}
\label{lem:proj_new}
Following  the assignment rule $\sigma$, we have that $\gamma_{\beta} \leq 1 - f(\alpha, \beta)$.
\end{claim}

For the sake of notation, we drop the subscript of $\gamma_{\beta}$ everywhere in the proof. 

\begin{proof}
Given $O$ and $B$, we consider a facility $o \in O$. Let $b\in B$  the closest facility to $o$, and $b'\in B$  the bicriteria solution that serves client $p$.  Now consider the $\alpha$-ball around $q$ and $\beta$-ball around $o$, for the ease of analyse we consider half plane above $x$-coordinate the same arguments hold for half plane below $x$-coordinate. We distinguish two cases.

\paragraph*{Case 1: $\ball(q, \alpha) \cap B \neq \varnothing$.} Suppose that $B$ contains a facility $b''$ lying inside the $\alpha$-ball around $q$.  Given that  $b\in B$ is the closest facility to $o$, it follows that $b''\in(\ball(q,\alpha)\backslash \ball(o, 1))$. In this case $\sigma(o)=\pi_{\text{cl}(B)}(o)$. Hence $\sigma(o)$ is no farther from $o$ than the facility $\pi_F((b+b'')/2)$ nearest to the midpoint of $b$ and $b''$.  Notice that the optimal center $o$ certifies the existence of a point in $F$ nearby the mid-point of $b$ and $b''$. The point $o\in F$ shows that $\pi_{F}((b+b'')/2)$ has distance at most $\alpha$ to $o$ because $||\pi_{F}((b+b'')/2)-(b+b'')/2||\leq ||(b+b'')/2-o||$. Hence we obtain $||\sigma(o)-o|| \leq \alpha$, therefore $ ||p-\sigma(o)|| \leq ||p-o||+ \alpha$ (see Figure \ref{fig:case1M}). Recall that the aim is to upper bound the displacement ratio \ref{def:gamma_beta} for client $p$, notice that $||p-o||+||p-b'|| \geq 1$, we obtain:
\begin{equation*}
     \begin{split}
     \gamma=\frac{||\sigma(o)-p||}{2||p-o|| + ||p-b'||} & \leq \frac{||p-o||+\alpha}{||p-o||+1}\\
     &\leq 1 - \frac{1-\alpha}{||p-o||+1}\\
     &\leq 1- \frac{1-\alpha}{2}
     \end{split}
\end{equation*}     

\begin{figure}[h]
\input{pictures}
     \centering
     \tikzV
      \caption{\footnotesize The midpoint of $b$ and $b''$ is shown by red dot, $||(b+b'')/2 - o || \leq \frac{\alpha}{2}$ and thus $||\sigma(o)-o||\leq\alpha$.}
        \label{fig:case1M}
\end{figure}    
\paragraph*{Case 2: $\ball(q, \alpha) \cap B = \varnothing$.} in the second case, where the $\alpha$-ball does not contain a facility from $B$, we argue that the points $o$, $\sigma(o) = b$, and $b'$ are far enough from a co-linear position. This allows us to argue that the triangle inequality in the proof of Lemma \ref{lem:projection} is not tight. Towards this, we divide the space into four regions that could contain client $p$, as indicated in Figure \ref{fig:regions}.

\begin{figure}[h]
\input{pictures}
     \centering
     \tikzIX
      \caption{ \footnotesize The dashed black circle depicts $\ball(o,1)$, while the dashed gray circles represent $\ball(o, 1-\omega)$ and $\ball(o,1+\omega)$. Regions $R_1$, $R_2$, $R_3$, and $R_4$ are outlined with green, yellow, purple, and blue borders respectively.}
        \label{fig:regions}  
\end{figure}
We define $q_1$ and $q_2$ as two points of intersection between $\ball(o,1)$ and $\ball (q, \alpha$). See Figure \ref{fig} for an illustration. We assume that $p$ lies the half plane above the $x$-axis. (The case where $p$ lies below the $x$-axis is symmetric.) Now, consider $q_3$ as the midpoint of $q$ and~$q_1$. Furthermore, define the region $H$ as the area above the lines passing through$(q_3 , o)$ and $(o,b)$. We define region $R_1 = H \backslash \ball(o , \beta)$.  Next, consider $(1-\omega)$ and $(1+ \omega)$ balls around $o$, $H'$ is defined as the area below the line passing through $(o, q_3)$ and above the line passing through $(o, q)$, we define $R_2 = (\ball(o, 1-\omega)\backslash \ball (o,\beta))\cap H'$, $R_3= (\ball(o, 1+\omega)\backslash \ball (o, 1-\omega))\cap H'$, and $ R_4=  H'\backslash \ball ( o , 1+ \omega)$, the regions are indicated in Figure \ref{fig:regions}.
\begin{itemize}
   \item $p \in R_1$, Let $b''$ be the closest point to $p$ not in the interior of $\ball(o,1)$, and let $p'$ be the point on the boundary of $\ball(o , \beta)$ that is closet to $p$. Let~$p''$ be the point where the segment $(o,q_3)$ intersects the boundary of $\ball(o,\beta)$, that is, $p''= (\beta \cos{\theta} , \beta \sin {\theta})$ where $\theta=\angle q_3 o q_1$. Notice that $\cos{\theta} = 1 - \frac{\alpha^2}{4}$, see Figure \ref{fig:a} for an illustration. First, we assume $p$ is inside $\ball(o,1+2\beta)$ in the region of $R_1$.

\begin{figure}[ht]
\input{pictures}
  \begin{subfigure}[b]{0.5\linewidth}
    \centering
    \tikzI
    \caption{\footnotesize  $p \in R_1$ } 
    \label{fig:a} 
    \vspace{4ex}
 \end{subfigure}
 \begin{subfigure}[b]{0.5\linewidth}
    \centering
    \tikzII    
    \caption{$p\in R_2$} 
    \label{fig:b} 
    \vspace{4ex}
 \end{subfigure} 
 \begin{subfigure}[b]{0.5\linewidth}
    \centering
    \tikzIII
    \caption{$p \in R_3$} 
    \label{fig:c} 
    \vspace{4ex}
 \end{subfigure}
 \begin{subfigure}[b]{0.5\linewidth}
 \centering
  \tikzIV
  \caption{$p\in R_4$}
  \label{fig:d} 
  \end{subfigure} 
  \caption{\footnotesize $o\in O$ is an optimum solution, $b\in B$ is the closest bicriteria solution to $o$, $1$-ball around $o$ is shown as a dashed circle, $\alpha$-ball around $q$ and $\beta$-ball around $o$ are shown in blue, $(1-\omega)$ and $(1+\omega)$ around $o$ are shown as blue dashed circles. The regions are specified by green borders. }
  \label{fig} 
\end{figure}

    \begin{observation} \label{obs:case2}
    For any $\epsilon_1 ,\epsilon_2 , X, Y \geq 0$ :  \[\frac{X-\epsilon_1+Y}{X+Y}\leq \frac{X-\epsilon_1 + Y+\epsilon_2}{X+Y+\epsilon_2}\]
    \end{observation}
Consider assigning $p$ via $p'$ to $b$. We bound the displacement cost as follows:
    \begin{equation*}
     \begin{split}
        \gamma_{\beta}= \frac{||\sigma(o)-p||}{2||p-o|| + ||p-b'||} & \leq \frac{||b-p'||+||p-p'||}{2||p'-o|| + ||p'-b''||+||p-p'||}\\
        &\leq \frac{||b-p''||+||p-p'||}{2\beta +1 -\beta+||p-p'||}\\
        & \leq \frac{\sqrt{(\beta \cos(\theta) + 1)^2+(\beta \sin(\theta))^2}+||p-p'||}{1 + \beta+||p-p'||}\\
        &=  \frac{\sqrt{\beta^2 + 2\beta \cos{\theta} +1}+||p-p'||}{1 + \beta+||p-p'||}\\
        &=  \frac{\sqrt{(1+\beta)^2 - \frac{\beta \alpha^2}{2}}+||p-p'||}{1 + \beta+||p-p'||}\\
     \end{split}
    \end{equation*}
    We assume $||p-p'||\leq 1+\beta$, and by observation \ref{obs:case2}, we obtain:
    \begin{equation*}
     \begin{split}
       \gamma_{\beta} &\leq \frac{ (1+\beta) (1+\sqrt{(1-\frac{\beta \alpha^2}{2}})}{2(1+\beta)} \leq \frac{1}{2} + \frac{\sqrt{1^2 - \frac{2\beta \alpha^2}{4} + \frac{\beta^2 \alpha^4}{16}}}{2}\leq\frac{1}{2}+ \frac{1- \frac{\beta \alpha^2}{4}}{2}\\
        &= 1 - \frac{4+\beta \alpha^2}{8}
     \end{split}
    \end{equation*}
Second, let's assume that the client $p$ is distant from $o$ and positioned within region $R_1$ outside  $\ball(o,1+2\beta)$, we can bound $\gamma_{\beta}$ as follows:
    \begin{equation*}
     \begin{split}
       \gamma_{\beta} &\leq  \frac{1 + ||o-p||}{2||o-p||}\leq \frac{1+1+2\beta}{2(1+2\beta)} = \frac{1+\beta}{1+2\beta} = 1 -\frac{\beta}{1+2\beta} 
     \end{split}
    \end{equation*}

 \item $p \in R_2$, we obtain the best location for $b'$ is when it lies on the point $q_1$ (see Figure \ref{fig:b} for an illustration), the cost of displacement is as follows:
    \begin{equation*}
     \begin{split}
        \gamma= \frac{||\sigma(o)-p||}{2||p-o|| + ||p-b'||} &\leq \frac{||p-o||+1}{||p-o|| + ||p'-o||+||p'-p||+||p-b'||}\\
        &\leq \frac{||p-o||+1}{||p-o||+ ||p''-o||+||p''-b'||}.
    \end{split}
    \end{equation*}
  To calculate $||p''-b'||$ we can consider  rotated $p''$ and $b'$ so that $p''=(\beta , 0)$ and $b'= (\cos\theta, \sin\theta)$, then:
    \begin{equation*}
       \gamma \leq \frac{||p-o|| +1}{||p-o|| + \beta + \sqrt{( \cos(\theta) - \beta)^2+(\sin(\theta))^2}}
    \end{equation*}
 Assume $\beta \leq \frac{\alpha}{12}$ then we have  $||p-o|| \leq 1-\beta$, hence:
    \begin{equation*}
        \begin{split}
         \gamma&\leq  \frac{2-\beta}{1+  \sqrt{\beta^2-2\beta(1-\frac{\alpha^2}{4})+1}}\\
        &\leq  \frac{2-\beta}{1 + (1-\beta)\sqrt{1+2\beta\frac{\alpha^2}{4}}}\\
        &=1-\frac{\sqrt{1+\beta\frac{\alpha^2}{2}}(1-\beta)-(1-\beta)}{1 + (1-\beta)\sqrt{1+\beta\frac{\alpha^2}{2}}}
        \end{split}
    \end{equation*}

    \item $p\in R_3$, we claim  the distance from $p$ to the closest bicriteria solution $b'$ can be bounded as $||p-b'|| \geq \frac{\alpha}{3}$.
    
    Suppose $\alpha\leq \frac{6}{10}$,  we define $z = 1 - ||q_3 - o|| \leq\frac{\alpha}{12}$, consider $p'''$ the closest point to $p$ on the line $(q_3, q)$,  hence $||p-p'''|| \leq \frac{\alpha}{6}$, now we have $\frac{\alpha}{2} \leq ||p'''-b'|| \leq ||p'''-p||+||p-b'||$ , then we obtain $||p-b'|| \geq \frac{\alpha}{3}$.
    \begin{equation*}
     \begin{split}
        \gamma= \frac{||\sigma(o)-p||}{2||p-o|| + ||p-b'||} & = \frac{||p-o||+1}{||p-o|| + (1- \frac{\alpha}{12})+ \frac{\alpha}{3}}\\
        &= \frac{||p-o|| +1}{||p-o|| + 1+ \frac{\alpha}{4}}\\
        &\leq \frac{2}{2+\frac{\alpha}{4}}\\
        &=1-\frac{\alpha}{8+\alpha}
     \end{split}
    \end{equation*}

    \item $p\in R_4$, assume $p$ and $b'$ are in the same location, therefore:
    \begin{equation*}
    \gamma= \frac{||\sigma(o)-p||}{2||p-o|| + ||p-b'||} \leq 1-\frac{||p-o||-||b-o||}{2||p-o|| }
    \end{equation*}
    Notice that $||p-o||\geq 1+ \frac{\alpha}{12}$, consider two cases either  $||p-o|| \leq 1+\alpha$ or  $||p-o|| > 1+\alpha$:
    \begin{itemize}
        \item  $||p-o|| \leq 1+\alpha$ 
        \begin{equation*}
            \gamma\leq 1 - \frac{{1+\frac{\alpha}{12}}-1}{2(1+\alpha)}  = 1 - \frac{\alpha}{24(1+\alpha)}
         \end{equation*}
        \item  $||p-o|| > 1+\alpha$ 
        \begin{equation*}
           \begin{split}
            \gamma\leq 1 - \frac{||p-o||-1} {2||p-o||}& \leq \frac{1}{2} +\frac{1}{2(1+\alpha)}\\
            &= 1- \frac{1+\alpha}{2(1+\alpha)}
            \end{split}
        \end{equation*}    
  
\end{itemize}
\end{itemize}
Therefore, by examining the position of $p$ in the regions, we establish that $\gamma_{\beta}$ is upper-bounded by $1 - f(\alpha, \beta)$. Consequently, Lemma \ref{lem:assignment} is substantiated by showing the existence of an $\alpha_0 \leq 0.6$ and a sufficiently small $\beta_0 \leq 0.05$ such that $\gamma_{\beta_0} \leq 1 - f(\alpha_0, \beta_0) = 1 - \epsilon_0 \leq 0.9978$. 
\end{proof}

In the proof of Theorem \ref{thm:disfkz}, we show that this new assignment property is enough to derive an improved \fpt approximation for \rkc in Euclidean space.
Since the assignment $\sigma$ maps every facility in $O$ uniquely to a facility in ${\sf cl}(B)$, this implies that $\sigma(O)$ is a feasible solution of cost at most $(3^z \cdot (1-\eta_0))\opt$. This certifies the existence of a feasible solution being a subset of ${\sf cl}(B)$ with the desired approximation factor. Hence, we can find such a solution in \fpt time by enumeration.

Now we are ready to prove Theorem \ref{thm:disfkz} (restated  for convenience).
\constapproxsfcz*
\begin{proof}
Let $B\subseteq {\mathbb R}^d$ denote a $(1 + \epsilon_0, k \cdot \ln n^2 / \epsilon_0^2)$ bi-criteria solution by applying the algorithm of \cite{goyal2021tight}. We denote the total cost of the set of $\beta_0$-near ($\beta_0$-far) points by $\bicn, \on, \an$ ($\bicf, \of, \af$) in a bi-criteria solution, an optimum solution, and the solution returned by our algorithm, respectively. Note that from the definition of $\beta_0$-near points and the (clustering)~cost, we have that $\on \leq \beta_0 \cdot \bicn \leq \beta_0 \cdot \opt$. Consequently, for the set of $\beta_0$-far points, we get $\of \geq (1-\beta_0)\cdot \opt$. 

For the set of $\beta_0$-near points, the best bound on the cost that our algorithm can achieve is that of Lemma~\ref{lem:projection}. Therefore, we get $\an \leq(2\on + \bicn)$. On the other hand, by Lemma~\ref{lem:assignment}, we do save some cost on the $\beta_0$-far points as we have $\af \leq \big((1-\epsilon_0)(2\of + \bicf)\big)$.
Moreover $\bicn+ \bicf = \bic \leq (1 + \epsilon_0)\opt$ by the choice of the bicriteria solution. Putting these factors together, we obtain

\begin{align*}
       \an+\af & \leq 3^{z-1} \cdot \big[ (2\on + \bicn)+(1 - \epsilon_0)^z(2\of + \bicf) \big]\\
     & \leq 3^{z-1} \cdot \big[ 2\beta_0^z \opt + 2 (1 - \epsilon_0)^z (1 - \beta_0^z) \opt + \bicn + (1 - \epsilon_0)^z \bicf \big]\\
     &\leq 3^{z-1} \cdot \big[ 2\beta_0^z \opt + 2 (1 - \epsilon_0)^z (1 - \beta_0^z) \opt + \bicn + \bicf  \big] \\ 
      & \leq 3^{z-1} \cdot \big[ 2\big( \beta_0^z  +  (1 - \epsilon_0)^z (1 - \beta_0^z) \big) \opt+  (1 + \epsilon_0) \opt \big]\\
       &\leq 3^{z-1} \cdot \big[\big( 2\beta_0^z  +  2(1 - \epsilon_0)^z -2(1-\epsilon_0)^z \beta_0^z +1 + \epsilon_0 \big)\opt \big]\\
         &\leq 3^{z-1} \cdot \big[\big( 2(1 - \epsilon_0)^z +2 \beta_0^z \epsilon_0 z +1 + \epsilon_0 \big)\opt \big]           \ \ \  \ \ \text{(using Bernoulli's  inequality\footnotemark )}\\
          &\leq 3^{z-1} \cdot \big[\underbrace{\big( 2(1 - \epsilon_0)^z +(2 \beta_0^z z +1) \epsilon_0+1 \big)}_{3 - \epsilon_1}\opt \big]\\
 &
\end{align*}
where in the first inequality, we used the approximate triangle inequality 
$(a + b + c)^z \leq 3^{z-1} \cdot (a^z + b^z + c^z)$ and in the second inequality, we use the fact that $\on + \of = \opt$ and thus, the maximizer for $2\on + 2(1 - \epsilon_0)^z \of$ occurs at the maximum of $\on$.
Next, to bound the above under-braced expression, we introduce a new parameter $\epsilon_1$, 
which we prove below in claim \ref{claim:kzs} that $\epsilon_1\geq 0.0018$. Finally, we have
\begin{align*}
    \an+\af &\leq 3^{z-1} \cdot (3-\epsilon_1)\opt  
           \\
          &\leq3^z \cdot (1 - \eta_0) \opt \ \text{where}\ \eta_0= \frac{\epsilon_1}{3}
          \\
\end{align*}

\begin{claim}
\label{claim:kzs}
 For any integer $z$, we have that $2(1-\epsilon_0)^z+\epsilon_0(1+2\beta_0^zz)\leq 2 -\epsilon_1$
\end{claim}
\begin{proof}
 Let $2(1-\epsilon_0)^z+\epsilon_0(1+2\beta_0^zz)\leq1.9982$.
Note that, since from Lemma~\ref{lem:assignment} we have that $\beta_0 = 0.05$ and $\epsilon_0>0.002$,  we get lower bound for $\epsilon_1$ as follows (the left-hand side is maximized at $z = 1$): $\epsilon_1 > 2- 2(1- \epsilon_0)- \epsilon_0(1+2\beta_0) >0.0018$.
\end{proof}
It remains to analyze the running time of the algorithm. In the initial phase of our algorithm, we invoke the Goyal-Jaiswal bi-criteria algorithm~\cite{goyal2021tight}. Subsequently, we evaluate all possible $k$-subsets of {\sf cl}(B), whose number is bounded by $\binom{(\lambda \cdot k )^2}{k} \leq (e \lambda)^{2k}$. This leads to an overall running time $\bigO {(e \lambda )^{2k}\cdot nk}$ where  $\lambda= \bigO{\frac{k}{\epsilon_0^2}\cdot \ln^2 n}$, therefore we have $\bigO{(e\frac{k}{\epsilon_0^2}\cdot \ln^2 n)^{2k} \cdot nk} = (\frac{k}{\epsilon_0})^{\bigO {k}} \cdot  n^{\bigO{1}} $. Let $\eta_1 = 3^z \cdot \eta_0$, Then for $\eta_0 > 0.0006$ and constant $z$, $\epsilon_1$ simplifies to $\frac{\eta_1}{3^{z-1}}$, resulting in an overall complexity of $(k/\epsilon_1)^{\bigO {k}} \cdot n^{\bigO{1}}$, which is \fpt in terms of $k$.

\end{proof}

\subsubsection{Proof of Claim~\ref{lem:restrict-to-plane}}\label{ss: lem plane}
We first define the ratio $\gamma'_\beta$:
\begin{equation*}\label{eq:displacement-ratio1}
 \gamma'_\beta = \max_{\substack{p\in \bR^d \backslash \ball(o, \beta),\\ b'\in S\setminus\ball(o,1) }} \left \{ \frac{||p-\sigma(o)||}{2||p - o|| + ||p - b'||} \right\}\,.
 \end{equation*}

Note that for the sake of the optimization problem, $b$ and $o$ are fixed points and therefore, the plane $S$ is a function of the location of point $p$. Hence, we write the plane $S$ as $S_p$.
We prove that $\gamma_\beta \leq \gamma'_\beta$. Let $b'$ and $p$ be the points optimizing $\gamma_\beta$. Suppose $b'$ is not on $S_p$.  We start by choosing an orthonormal basis $(x,y,z)$ for the linear space spanning $b$, $o$, $p$ and $b'$, and fixing a system of coordinates. Towards this, let $z = b-o$, $x = \frac{p-\hat{p}}{||p-\hat{p}||}$, where $\hat{p}$ is the orthogonal projection of $p$ on the line containing $b$ and $o$. Then fix $y = \frac{b'-\hat{b'}}{||b'-\hat{b'}||}$, where $\hat{b'}$ is the orthogonal projection of $b'$ on S,  the plane containing $b$, $o$, and $p$. Now we fix the origin of the coordinate system to be in the center of the following disk \footnote{As a clarification, we temporarily relocate the origin of the coordinate system to the center of the disk and we fixed $q=(0,0,q_z)$.}.

Let $D$ be the disk the perimeter of which is defined as the circle $\partial(\ball(o,1)) \cap \partial(\ball(q, \alpha))$, where $\partial$ indicates the boundary of the closed space. Notice that after this translation of the coordinate system, $D$ is contained in the $x$-$y$ plane and that  $q,p$, and $b'$ can be represented as $q = (0,0, q_z)$, $p =(p_x,0,p_z)$, and $b' =(b'_x , b'_y  , b'_z )$.

For any point $s=(s_x,s_y,s_z)$ let $\bar{s}=(s_x,s_y,0)$ denote its projection onto the $x$-$y$ plane. See Figure~\ref{fig:plane} for an illustration of the projections $\bar{q}$, $\bar{p}$, and $\bar{b'}$ of $q,p$, and $b'$ on the $x$-$y$ plane, respectively.

\begin{figure} [ht]
    \input{pictures}
    \centering
    \tikzVIII
    \caption{\footnotesize {An illustration of $D$ and the projection points of $p$, $q$ and $b'$.}}
    \label{fig:plane}
\end{figure}
\begin{proposition} \label{cl:rotation}
    If $b'_y \neq 0$ then there exists $b'' \notin \interior(\ball (o,1))$ so that $||b''-p||^2 < ||b' - p||^2$ and so that $b''\in \interior(\ball (q, \alpha))$ if and only if $b'\in\interior(\ball (q, \alpha))$.
\end{proposition}
\begin{proof}
Consider $ b'' $ as a rotation of $ b' $ with respect to the line $(o,q)$, i.e., a rotation that preserves  disk $D$. We choose to rotate by the angle that will place $\overline{b''}$ as close to $\overline{p}$ as possible, i.e., $\overline{b''}$ will be co-linear with $\overline{q}$ and $\overline{p}$ ($\overline{b''}_y = \overline{q}_y = \overline{p}_y = 0$). Note that the rotation preserves the distances $||b'' - o||$ and $||b'' - q||$ to be equal to $||b' - o||$ and $||b' - q||$. To see that, assume that $\check{b'}$ is the orthogonal projection of $b'$ on the line containing $o$ and $q$ We have that $||b' - o|| = \sqrt{(||b' - \check{b}||)^2 + (||\check{b'} - o||)^2}$. The claim follows since the orthogonal projection of $b''$ on the line would also land on $\check{b'}$. The argument for $||b'' - q||$ follows in a similar way. As a consequence, since $b'$ is outside $\ball (o, 1)$, then $b''$ is also outside $\ball (o, 1)$. Also, $b''$ will fall outside of $\ball(q, \beta)$ if and only if $b'$ is outside $\ball(q, \beta)$.

We observe that $||b'_{z}-p_z||^2$  is constant w.r.t. rotations. This allows us to reduce our argument to analyzing the change of squared distances within the $x$--$y$ plane.

\begin{align*}
    ||b'-p||^2 & = ||b'_x-p_x||^2+ ||b'_{y}-p_y||^2 + ||b'_{z}-p_z||^2 \\
               & = ||\overline{b'} - \overline{p}||^2 + ||b'_{z}-p_z||^2 \\
               & > ||\overline{b''} - \overline{p}||^2 + ||b'_{z}-p_z||^2 \\
               & = ||b''_x-p_x||^2+ ||b''_{y}-p_y||^2 + ||b'_{z}-p_z||^2 \\
               & = ||b''_x-p_x||^2+ ||b''_{y}-p_y||^2 + ||b''_{z}-p_z||^2 \\
               & = ||b''-p||^2.
\end{align*}

It remains to observe that $||b''-p||^2 < ||b' - p||^2$ implies $||b''-p|| < ||b' - p||$.

\end{proof}

By Proposition \ref{cl:rotation}, we obtain that $\overline{b'}$ must be co-linear with $\overline{p}$ and $\overline{q}$, therefore $\gamma_\beta \leq \gamma'_\beta$ and point $b'$ may be assumed to be on the same plane $S$ as points $b$, $o$, and $p$. The fact that $b'' \notin \interior(\ball (o,1))$ ensures $b$ remains the closest point in $F$ to $o$ when replacing $b'$ with $b''$. The property that $b''\in \interior(\ball (q, \alpha))$ if and only if $b'\in\interior(\ball (q, \alpha))$ ensures that, after replacing $b'$ with $b''$, $o$ is assigned to the midpoints of $b''$ and $b$ if and only if we it was assigned to the midpoint of $b'$ and $b$ before the replacement. The two properties together guarantee that the replacement of $b'$ with $b''$ does not change $\sigma(o)$.

\subsection{Hardness of Discrete \kc}\label{sec:hardness-discrete}

For this section, we use the following explicit construction of the so-called \emph{$\eta$-balanced error-correcting codes} from a recent result of Ta-Shma~\cite{ta-shma17:explicit-balanced-codes} which we rephrase for our purposes as follows:

\begin{theorem}\label{thm:balanced-code}
Let $\eta\in (0,1/2)$ be a positive constant. Then there is an algorithm that computes, for any given number $s\in\bN$, an $s$-element set $B\subseteq\{0,1\}^t$ of binary vectors of dimension $t=\OO(\log s/\eta^{2+o(1)})$ such that for any $b\in  B$, its Hamming weight $||b||_1$ and for any $b'\in B\setminus\{b\}$, the Hamming distance $||b-b'||_1$ both lie in the interval $[(1/2-\eta)t,(1/2+\eta)t]$. The running time of the algorithm is $\OO(st)$.
\end{theorem}
\begin{proof}
    Ta-Shma~\cite{ta-shma17:explicit-balanced-codes} gives an explicit construction of a $t\times \lceil \log_2 s\rceil$ binary matrix generating a linear, binary, error-correcting code of message length $\lceil\log_2 s\rceil$, block length $t=\OO(\log s/\eta^{2+o(1)})$, and pairwise Hamming distance between $(1/2-\eta)t$ and $(1/2+\eta)t$. Since the code is linear, it contains the zero code word. Hence each code word has Hamming weight in $[(1/2-\eta)t, (1/2+\eta)t]$. The time for constructing the matrix is polynomial in $\log s$ and $t$.  Using the generating matrix, at least $s$ many non-zero code words can be enumerated in time $\OO(st)$, which dominates the time for computing the matrix.

\end{proof}

We leverage balanced error correcting codes as gadget in our hardness proof for discrete \kc . For any binary vector $b\in\{0,1\}^t$, we denote by $\bar{b}$ the binary vector obtained by flipping each coordinate in $b$.

\hardnesseuc*

\begin{proof}
We show a reduction from \mcis, which is known to be $W[1]$-hard~\cite{CyganFKLMPPS15}. The input is a $k$-partite graph $G=(V,E)$ with $k$-partition $V_1,\dots,V_k$. The question is if there is an independent set that is \emph{multi-colored}, that is, it has precisely one node from each set $V_i$, $i\in[k]$. W.l.o.g.\ we assume that each $V_i$ contains at least one node that is adjacent to all nodes $V\setminus V_i$. Adding such nodes, we can additionally assume that $|V_i|=n/k$ for each $i\in[k]$ where $n=|V|$.

Fix some constant $\eta\in(0,1/2)$. Using Theorem~\ref{thm:balanced-code}, we construct a set $B\subseteq\{0,1\}^t$ of $n$ nearly equidistant code words of dimension $t=\OO(\log n / \eta^{2+o(1)})$. We map each node $u\in V$ uniquely to some non-zero code word $b(u)\in B$. 
We construct a \kc instance in $\bR^{k\cdot t}$ as follows. We subdivide the coordinates of each point in $\bR^{k\cdot t}$ into $k$ \emph{blocks} each containing $t$ consecutive coordinates. In our set $P$ of data points, we introduce for each node $v_i\in V_i$, $i\in [k]$, the point $p(v_i)\in P$ in which the $i$th block equals $b(v_i)$ and all other coordinates are zero. For each edge $(v_i,v_j)\in E$, $v_i\in V_i$, $v_j\in V_j$ for distinct $i,j\in[k]$ we create a point $p(v_i,v_j)\in P$ in which the $i$th block equals $\overline{b(v_i)}$, the $j$th block equals $\overline{b(v_j)}$, and all other coordinates are zero. No further points are added to~$P$. We set the number of centers to be $k$ completing the construction of the \kc instance.

Let $i\in [k]$ and $v_i,v_i'\in V_i$ be distinct vertices. We have that $||p(v_i)-p(v_i')||_q^q\leq ||b(v_i)-b(v_i')||_1\leq (1/2+\eta)t$ by Theorem~\ref{thm:balanced-code}. Let $v_j\in V_j$, $j\in [k]$ such that $(v_i,v_j)\in E$. By Theorem~\ref{thm:balanced-code}, we have that 
\begin{align*}
 ||p(v_i')-p(v_i,v_j)||_q^q & \leq ||b(v_i')-\overline{b(v_i)}||_1
+||\overline{b(v_j)}||_1\\
  & \leq (t-||b(v_i')-b(v_i)||_1)+(t-(1/2-\eta)t)\\
  & \leq (t-(1/2-\eta)t)+(1/2+\eta)t\\
 & \leq (1+2\eta)t\,.
\end{align*}
Hence if there is a multi-colored independent set $I$ for $G$ then $X=\{\,p(u)\mid u\in I\,\}$ is a $k$-element set such that $\delta(p,X)^q\leq (1+2\eta)t$ for any $p\in P$ under the $\ell_q$ metric, which gives an upper bound of $(1+2\eta)t$ on the \kc objective in the completeness case.

For analyzing the soundness case,  assume that there is no multi-colored independent set for $G$. Consider an arbitrary $k$-element set $X\subseteq V$. We say that $x\in X$ \emph{covers} $p\in P$ if $\delta(p,x)^q < (3/2-3\eta)t$. We claim that there is some~$p\in P$ not covered by any center in~$X$. The correctness of this claim implies that any parameterized approximation algorithm with approximation ratio strictly better than $((3/2-3\eta)/(1+2\eta))^{1/q}$ implies that $W[1]=\fpt$ and thus the theorem.

In order to prove this claim, we assume for the sake of contradiction, that all $p\in P$ are covered by some center in $X$. First, we argue that w.l.o.g.\ $X$ contains no point of the form $p(v_i,v_j)$ where $(v_i,v_j)\in E$. In fact, for any $g\notin\{i,j\}$, we have that
\begin{align}\label{eq:vert-edge-code}
\begin{split}
 ||p(v_g')-p(v_i,v_j)||_q^q &\geq
 ||b(v_g')||_1+||\overline{b(v_i)}||_1+||\overline{b(v_j)}||_1\\
 & \geq (1/2-\eta)t+2(t-(1/2+\eta)t)\\
 & = (3/2-3\eta)t\,.
 \end{split}
\end{align}
Hence $p(v_i,v_j)$ can cover $p(v_g')$ only if $g=i$ or $g=j$. Similarly, $p(v_i,v_j)$ can cover $p(v_g',v_h')$ only if $i=g$ and $j=h$. But then these points would be covered by $p(v_i)$ as well and hence we could replace $p(v_i,v_j)$ with $p(v_i)$. We therefore assume that $X$ contains only points of the form $p(v_i)$.

We claim that $X$ is multi-colored. Otherwise, there would be some $V_i$ that contains no point from $X$. By our initial assumption, $V_i$ contains some point $v_i$ that is adjacent to all points $V\setminus V_i$. Assuming $k\geq 3$ there exists at least one $V_j$, $j\neq i$ that contains at most one node from $X$. If $V_j$ intersects $X$ then let $v_j\in V_j \cap X$ , and otherwise let $v_j$ be an arbitrary node in $V_j$. By our assumption $(v_i,v_j)\in E$. If $v_j\in X$ then 
\begin{align}\label{eq:incident-code}
\begin{split}
 ||p(v_j,v_i)-p(v_j)||_q^q & \geq ||\overline{b(v_j)}-b(v_j)||_1+||\overline{b(v_i)}||_1\\
 & \geq t+(t-(1/2+\eta)t)\\
 & = (3/2-\eta)t
\end{split}
\end{align}
 as the $j$th block of $p(v_j)$ equals $b(v_j)$ and the $i$th block of $p(v_j,v_i)$ equals~$\overline{b(v_j)}$. If $v_j\notin X$ then for any $v_h\in X$ we have $h\notin\{i,j\}$. 
Thus $||p(v_i,v_j)-p(v_h)||_q^q\geq (3/2-3\eta)t$, which follows as in~(\ref{eq:vert-edge-code}). Hence $p(v_i,v_j)$ would not be covered showing that $X$ is multi-colored.
Since $X$ is multi-colored it can not be an independent set. Hence there exists some edge $(v_i,v_j)$ such that $v_i,v_j\in X$ but then $||p(v_i)-p(v_i,v_j)||_q^q\geq (3/2-\eta)t$, $||p(v_j)-p(v_i,v_j)||_q^q\geq (3/2-\eta)t$, and $||p(v_h)-p(v_i,v_j)||_q^q\geq (3/2-3\eta)t$ for any $v_h\in X$, $h\notin\{i,j\}$, which follows as in~(\ref{eq:incident-code}) and~(\ref{eq:vert-edge-code}), respectively. Hence $\delta(p(v_i,v_j),X)\geq (3/2-3\eta)t$, implies that $p(v_i,v_j)$ is not covered.

We complete the proof by noting that the dimension of the instance can be reduced to $O(\log n)$ for Euclidean metrics by using the Johnson-Lindenstrauss transform with sufficiently small (constant) error parameter.
\end{proof}

\section{EPAS for Metrics of Sub-Logarithmic Doubling Dimension}\label{sec:paskd_sfc}

In this section, we show an EPAS for \rkc in  metrics of  sub-logarithmic doubling dimension. This result complements the hardness result of Section~\ref{sec:discrete} (Theorem~\ref{thm: LB-intro}). Towards our goal, we prove the following result. 

\epasdoubling*

Note that the above algorithm runs in \fpt time for $d=o(\log n)$.
We also remark that the above result can be extended to the continuous $\mathbb{R}^d$.
Throughout this section, we assume that the weight aspect ratio $\frac{\max_{p \in P} w(p)}{\min_{p' \in P} w(p')}$ and the distance aspect ratio $\frac{\max_{p,p' \in P} \delta(p,p')}{\min_{p\ne p' \in P} \delta(p,p')}$ are bounded by $\poly(n)$, some polynomial in $n$.
For $p\in P$ and any number $r \ge 0$, denote by $\ball(p,r)$ to be the closed ball centered at $p$ of radius $r$.
We prove the theorem in two steps: first, in Section~\ref{ss:coreset_sfc} we show an algorithm to obtain a  coreset for the problem, and then, in Section~\ref{ss:pas_kd_sfc} we show how to use this coreset to get the algorithm of Theorem~\ref{thm:pas_kd}.

\subsection{Coreset for \rkc} \label{ss:coreset_sfc}
The key idea for constructing coresets for \rkc crucially relies on the following alternate but equivalent definition of the problem. In this definition, we are given $\cI=(F, P\subset \cM,\cW)$, where either $F=\cM$ or $F\subseteq \cM$, where $\cM$ is doubling metric of dimension $d$, defined by the metric function $\delta$.
A \emph{group} is a  weight vector $\vecw \in \cW$ such that $\vecw: P \rightarrow \mathbb{R}_{\geq 0}$. 
Given $X \subseteq F$, the distance vector $\vecd_{P}(X)$ is defined as ${\vecd }_{P}(X)[p] = \delta(p,X)^z$, for each $p \in P$. 
The cost of $X$ for a group $\vecw \in \cW$ is defined as $c(\vecw, X) = \vecw \cdot \vecd_P(X)$.
For  a \rkc instance $\cI= (F,P,\cW)$, the cost of $X$ is defined as $\cost(\cI,X) = \max_{\vecw \in \cW}\cost(\vecw,X)$.
The cost of the  instance $\cI= (F,P,\cW)$ is 
\[\opt(\cI) = \min_{X \subseteq F, |X|=k} \max_{\vecw \in \cW} \cost(\vecw,X) 
\] 
Whenever the instance $\cI$ is clear from context, we will just write $\opt$. 
Notice that, in the original \rkc, a group is given by $S \subseteq P$, and this can be captured by weight vector $\vecw[p] = 0$ for $p \not \in S$ and $w(p)$ otherwise. We prove the following coreset exists for  \rkc.

\begin{restatable}[Coreset for \rkc]{theorem}{coresetrkc}
\label{thm: coresetrobust}
Given an instance $\cI=(F,P,\cW)$ of \rkc in doubling metric of dimension $d$ and $0 < \epsilon \le 1$, there is an algorithm that, in time $\left(\frac{2^z}{\epsilon}\right)^{\bigO{d}}\poly(n,m)$, computes another instance  $\cI'=(F,P',\cW')$ of \rkc with $P' \subseteq P: |P'|= \left(\frac{2^z}{\epsilon}\right)^{\bigO{d}} kz\log n$ such that for any $X \subseteq F$ with $\abs{X} = k$,
\[(1-\epsilon) \cost(\cI,X) \le \cost(\cI',X)  \le (1+\epsilon) \cost(\cI,X). \]

\end{restatable}

We remark that the above theorem yields a coreset of clients, and not of groups, and hence, the total size of coreset is comparable to the original instance. 
However, we will show later that such coreset is sufficient to get a parameterized approximation scheme with parameters $k$ and $d$.
We would also like to point out that the exponential dependency on $d$ on the point set size of the coreset is inevitable since \rkc captures \kc, for which such a lower bound is known~\cite{braverman2019coresets, pmlr-v119-baker20a}. To see that our notion of coreset for \rkc coincides with the regular notion of coreset for \kc, note that in this setting each group contains a single distinct point.

In the next section, we describe the algorithm of Theorem~\ref{thm: coresetrobust}.

\paragraph*{The Algorithm.}

Our algorithm (See Algorithm~\ref{algo:coreset}) is inspired by the grid construction approach of~\cite{10.1145/1007352.1007400} that yields coresets for \km\ and \kmns. 
Given an instance $\cI=(F,P,\cW)$ of \rkc, the first step is to start with an $(\alpha,\beta)$-bicriteria solution $B=\{b_i\}_{i \in [\beta k]}$ that opens at most $\beta k$ facilities with the guarantee that $\cost({\cI},B) \le \alpha \cdot \opt$, for some constants $\alpha,\beta \ge 1$. Let $R = \sqrt[z]{\frac{\cost({\cI},B)}{\alpha \tau}}$, where $\tau := \max_{\vecw \in \cW} ||\vecw||_1$.
Let $\Delta = \frac{\max_{p \in P, \vecw \in \cW} \vecw[p]}{\min_{p \in P, \vecw \in \cW} \vecw[p]}$ be the weight aspect ratio of $\cI$. 
Then, for each $b_i \in B$, 
consider the balls $\cB^j_i := \ball(b_i,2^j R)$, for $j \in \{0,\cdots, \lceil 2 \log(\alpha n \Delta) \rceil\}$. Note that, for $\vecw \in \cW$ and $p\in P$ with $\vecw[p]>0$, it holds that $\delta(p,B) \le R\sqrt[z]{\alpha n \tau}$, since $\delta(p,B) \le \sqrt[z]{\frac{\cost({\cI},B)}{\vecw[p]}} \le \sqrt[z]{\frac{\alpha\tau}{\vecw[p]}}R \le R\sqrt[z]{\alpha n \Delta}$.
Hence, we have that every point $p \in P$ is contained in some ball $\cB^j_i$. 
For $b_i \in B$, let $\cQ^j_i = \cB^j_i - \cB^{j-1}_i$, for $j= \{1,\cdots, \lceil 2 \log(\alpha \Delta) \rceil\}$, be the ring between $\cB^j_i$ and $\cB^{j-1}_i$, with $\cQ^0_i = \cB^0_i$.
Decompose every ball $\cB^j_i$ into smaller balls each of radius $\frac{\epsilon }{40 \alpha}R 2^j$ using the fact that the metric is a doubling metric.
These balls can intersect, so we assign  every point $p \in P$ to exactly one ball (for example, by associating $p$ to the smallest ball containing $p$, breaking ties arbitrarily).\\
For every ball $\cB^j_i$ and every smaller ball $t$ of $\cB^j_i$ with $|t \cap \cQ^j_i| \ne \emptyset$, pick an arbitrary point $p' \in t \cap \cQ^j_i$ as the \emph{representative} of (the points in) $t\cap \cQ^j_i$, and add $p'$ to the coreset $P'$ with group weight vectors as follows. Corresponding to every group vector $\vecw \in \cW$, create a new group vector $\vecw' \in \cW'$. Then,$
\vecw'[p'] := \sum_{p \in t \cap Q^j_i} \vecw(p)$.
Intuitively, $\vecw'[p']$ captures the total weight of points of $\vecw$ in $t \cap \cQ^j_i$. This concludes the coreset construction.

The high-level idea above is to decompose each ball $\cB^j_i$ into smaller balls and pick a distinct point as the representative of points in the non-empty decomposed ball. Additionally, such representative $p'$ participates in the group $\vecw'$ with weight which is  sum of the weights of  points in $\vecw$ that are represented by $p'$. 
However, we want to decompose the ball $\cB^j_i$ into smaller balls in a way  that the total number of balls remains the same, irrespective of the radius of the ball. This is necessary as for higher values of $j$, this number would depend on $n$, if we are not careful. While this does not seem to help much, as the radius of the decomposed balls is much large for higher values $j$, it actually does the trick:  since the points in these balls are far from $b_i$, and hence their connection cost to $b_i$ is also large. This allows us to represent the radii of larger balls in terms of the connection cost of its points to $B$, thus bounding the error in terms of the cost of $B$, which in turn is bounded by $\alpha \opt$, which gives us the desired guarantee.

\begin{algorithm}[t]
\caption{Coreset construction for \rkc} \label{algo:coreset}
  \KwData{Instance $\cI=(F,P,\cW)$ of \rkc, $(\alpha,\beta)$-bicriteria solution $B$ for $\cI$}
  \KwResult{Coreset $\cI'=(F,P',\cW')$  for $\cI$}
    Let $P' \leftarrow \emptyset$\;
    Let $\vecw'   \leftarrow \textbf{0}$ for $\vecw' \in \cW'$\;
    Let $\tau \gets \max_{\vecw \in \cW} ||\vecw||_1$\;
    Let $R= \sqrt[z]{\frac{\cost({\cI},B)}{\alpha \tau}}$\;
    For each $b_i \in B$ for $j \in\{0,1,\cdots, \lceil 2\log \alpha n \Delta \rceil \}$, let $\cB^j_i = \ball(b_i,2^jR)$, and let $\cQ^j_i = \cB^j_i - \cB^{j-1}_i$  with $\cQ^0_i=\cB^0_i$\;
    
    Decompose each ball $\cB^j_i$ into balls of radius each $\frac{\epsilon}{\alpha 3^{z+2}}2^j R$ \tcp*{e.g., use Lemma~\ref{apd:balldecomp}}
    
    Associate each point $p \in P$ to a smallest ball containing $p$ breaking ties arbitrary\;
    
    \ForEach{$i\in[k]$}{
        \ForEach{ $j \in\{0,1,\cdots, \lceil 2\log \alpha n \Delta \rceil \}$}{
            \ForEach{smaller ball $t$ of $\cB^j_i$}{
                \If{$\exists p \in t \cap Q^j_i $}{
                     $P' \gets P' \cup p$ \;
                     \ForEach{$\vcw \in \cW$}{Set the corresponding weight vector $\vcw'[p] = \sum_{p' \in t \cap \cQ^j_i} \vcw[p']$\;}               
                     \textbf{break}\;
                }
            }
          }
     }
    \Return{$\cI':=(F,P',\cW')$\;}
\end{algorithm}

\paragraph*{Analysis.}
First, let us bound $|P'|$. 
Note $|P'|= \OO(|B|\log(\alpha n) \left(\frac{\alpha 3^z}{\epsilon}\right)^{\bigO{d}})$, assuming $\Delta = \poly(n,m)$.
We will use the following bi-criteria algorithm for \rkc due~\cite{pmlr-v134-makarychev21a}.
\begin{theorem}[~\cite{pmlr-v134-makarychev21a}]\label{thm:res:bic}
There exists a polynomial time algorithm for \rkc  that, for every $\gamma \in (0,1)$, outputs at most $\frac{k}{1-\gamma}$ centers whose cost is bounded by $\frac{e^{O(z)}}{1-\gamma}$ times the optimal cost. 
\end{theorem}
Invoking Algorithm~\ref{algo:coreset} with the above bi-criteria solution for $\gamma=1/2$, we get $|P'|= \left(\frac{2^z}{\epsilon}\right)^{\bigO{d}} kz\log n$, as desired. Note that the running time of the algorithm is $\left(\frac{2^z}{\epsilon}\right)^{\bigO{d}}\poly(n)$.

 We now argue that $\cI'$ is indeed a coreset for $\cI$.
    Let $\opt(\cI)$ be the cost of optimal solution for $\cI$.
    Fix a feasible solution $X \subseteq F, |X|=k$ and let $\hat{\vcw} \in \cW$ be a maximizer of the \rkm cost of $\cI$ for $X$. We claim that, for any $\vcw' \in \cW'$, it holds that
       \[
                \cost(\vcw,X) - \epsilon \cost(\hat{\vcw},X) \le \cost(\vcw',X) \le (1+\epsilon) \cost(\hat{\vcw},X),
    \]
    where  $\vcw \in \cW$ is the corresponding weight vector to $\vcw'$. Fix any $\vcw' \in \cW'$ and the corresponding $\vcw \in \cW$.
For $p \in P$, let $r(p)\in P'$ be the representative of $p$.
Using the inequality\footnote{This can be proved using induction on $z$.}  $|a^z-b^z| \le |(a-b)(a+b)^{z-1}|$, for $a,b\ge 0$, we have that the total error $\cE:=|\cost(\vecw',X) -\cost(\vecw,X)|$ is bounded by, 
\[
\cE  \le \sum_{p \in P} |(\vecw[p]\delta(p,X)^z - \vecw[p]\delta(r(p),X)^z)|\le \sum_{p \in P} \vecw[p]|(\delta(p,X) - \delta(r(p),X))(\delta(p,X) + \delta(r(p),X))^{z-1}|
\]
Note that $|(\delta(p,X) - \delta(r(p),X))| \le |(\delta(p,r(p)) + \delta(r(p),X) - \delta(r(p),X))| \le \delta(p,r(p))$. Further, $\delta(p,X) + \delta(r(p),X) \le 2\delta(p,X) + \delta(p,r(p))$. Hence,
\[
\cE \le \sum_{p \in P} \vecw[p] \cdot \delta(p,r(p)) (2\delta(p,X) + \delta(p,r(p)))^{z-1}.
\]
To bound $\cE$, we divide the points in $P$ in three parts, and bound the errors on each part separately. Let $P_R:=\{p \in P \mid \delta(p,B) \le R \ \&\  \delta(p,X)\le R\}$, $P_B:= \{p \in P \mid \delta(p,B) > R \text{ and } \delta(p,X) \le \delta(p,B) \}$,  and $P_X:= \{p \in P \mid \delta(p,X) > R \text{ and } \delta(p,B) \le \delta(p,X)\}$.

For $p \in P_R$, we have $\delta(p,r(p)) \le \frac{\epsilon R}{\alpha 3^{z+1}} $, and $\delta(p,X)\le R$, and hence we have,
\begin{align*}
\cE_R &\le \sum_{p\in P_R} \vecw[p] \frac{\epsilon R}{\alpha  3^{z+1}}\left(2R+\frac{\epsilon R}{\alpha  3^{z+1}}\right)^{z-1} \\
& \le \frac{\epsilon R^2}{\alpha  3^{z+1}} \left(2+\frac{\epsilon}{\alpha  3^{z+1}}\right)^{z-1} \sum_{p\in P_R} \vecw[p] \\
&\le \frac{\epsilon \opt}{9 \alpha\tau} \sum_{p\in P_R} \vecw[p] \\
&< \frac{\epsilon }{3}\opt \qquad \text{since } \tau \ge ||\vecw||_1. \\
\end{align*}
For $p \in P_B$, suppose  $2^j R \ge \delta(p,B) > 2^{j-1} R$, for some $j\ge 1$. Then, note that
$
\delta(p,r(p)) \le 2\frac{\epsilon 2^{j} R}{\alpha  3^{z+2}} < 2\frac{2\epsilon \delta(p,B)}{\alpha 3^{z+2}} < \frac{\epsilon}{\alpha  3^z} \delta(p,B)$. 
Hence, we bound $\cE_B$ using  the fact $\delta(p,X) \le \delta(p,B)$,
\begin{align*}
    \cE_B &\le \sum_{p\in P_B} \vecw[p] \frac{\epsilon}{\alpha  3^z} \delta(p,B)^z \left(2+ \frac{\epsilon}{\alpha 3^z} \right)^{z-1}
    \le \frac{\epsilon}{3\alpha} \sum_{p\in P_B} \vecw[p]  \delta(p,B)^z
    \le  \frac{\epsilon}{3}\opt(\cI).
\end{align*}
For $p \in P_X$, suppose $\delta(p,B) > R$, then $2^j R \ge \delta(p,B) > 2^{j-1} R$, for some $j\ge 1$. In this case, we have
$
\delta(p,r(p)) \le 2\frac{\epsilon 2^{j} R}{\alpha  3^{z+2}} < 2\frac{2\epsilon \delta(p,B)}{\alpha 3^{z+2}} < \frac{\epsilon}{\alpha  3^z} \delta(p,B) \le \frac{\epsilon}{\alpha  3^z} \delta(p,X)$. Otherwise, $\delta(p,B) \le R$, in which case $\delta(p,r(p)) \le 2\frac{\epsilon R}{\alpha  3^{z+2}} < \frac{2\epsilon}{\alpha 3^{z+2}} \delta(p,X)$.
Hence, 
\begin{align*}
    \cE_X &\le \sum_{p\in P_X} \vecw[p] \frac{\epsilon}{\alpha 3^z} \delta(p,X)^z \left(2+ \frac{\epsilon}{\alpha 3^z} \right)^{z-1} 
    \le \frac{\epsilon}{3\alpha} \sum_{p\in P_X} \vecw[p]  \delta(p,X)^z
    \le  \frac{\epsilon}{3}\cost(\vecw,X).
\end{align*}
Now, 
   \begin{align*}
       \cost(\vcw',X) &\le \cost(\vcw,X) + \epsilon \opt(\cI)\\
       &\le  \cost(\vcw,X) + \epsilon \cost(\hat{\vcw},X) \qquad \text{since } \opt(\cI) \le \cost(\cI,X) = \cost(\hat{\vcw},X)\\
       &\le (1+\epsilon) \cost(\hat{\vcw},X) \qquad \text{since } \cost(\vcw,X) \le \cost(\hat{\vcw},X).
   \end{align*}
   Similarly,       $\cost(\vcw',X) \ge \cost(\vcw,X) - \epsilon \opt(\cI)
       \ge  \cost(\vcw,X) - \epsilon \cost(\hat{\vcw},X)$  since  $\opt(\cI) \le  \cost(\hat{\vcw},X)$.
    
    Now, we finish the proof as follows.
    \[
                \cost(\cI',X) = \max_{\vcw' \in \cW} \cost(\vcw',X) \le (1+\epsilon) \cost(\hat{\vcw},X) = (1+\epsilon) \cost(\cI,X)
    \]
    On the other hand,
    \[
             \cost(\cI',X) = \max_{\vcw' \in \cW} \cost(\vcw',X) \ge \max_{\vcw \in \cW} \cost(\vcw,X) - \epsilon \cost(\hat{\vcw},X) = (1-\epsilon) \cost(\cI,X).
    \]

\subsection{EPAS for \rkc} \label{ss:pas_kd_sfc}
In this section, we show how to use the coreset obtained from Theorem~\ref{thm: coresetrobust} to get a $(1+\epsilon)$-approximate solution to the \rkc problem and provide an EPAS with respect to $k$ and $d$, when $|P|$ is small. 
By scaling the distances in the instance of \rkc, we assume that the distances are between $1$ and $\Delta'$, for some number $\Delta'$.
Our algorithm (see Algorithm~\ref{alg:epas_kd}) uses the leader guessing idea of~\cite{cohenaddad_et_al:LIPIcs:2019:10618}. In the leader guessing approach, we guess the leader of every partition of a fixed optimal solution, where the leader of a partition is a closest point (client) in $P$ to the corresponding optimal center. However, each point can participate in multiple groups, resulting in the total number of points being dependent on the number of groups, $|\cW|$.
But, we will show next that guessing the leaders from $P$ without considering the groups in $\cW$ is, in fact, sufficient. Further, to get a $(1+\epsilon)$-approximate solution, we use a standard ball decomposition lemma (e.g., use Lemma~\ref{apd:balldecomp}).

\begin{restatable}{theorem}{leaderelection}\label{thm:pas_kdp}
For any $0< \epsilon \le 1$, Algorithm~\ref{alg:epas_kd}, on input $\cI= (F,P,\cW)$, computes $X \subseteq F: |X| \le k$ such that $\cost({\cI},X) \le (1+ \epsilon)\opt(\cI)$
in time $\left((\frac{z}{\epsilon})^d \log n \right)^{\bigO{k}}|P|^{k}   \poly(n,m)$.

\end{restatable}
\begin{proof}
First we bound the runtime of the algorithm. The leader enumeration (first forall loop) requires at most  $|P|^k$ loops, one for each $k$-tuple of $P$. 
Assuming $\Delta' = \poly(n)$,
the radii enumeration (second forall loop) for the leaders requires at most $\left(\frac{\nicefrac{\log n}{\epsilon}}{\epsilon}\right)^{\bigO{k}} = (\log n/\epsilon)^{\bigO{k}} $ loops using discretized steps of size that is power of $(1+\nicefrac{\epsilon}{10z})$.
Finally, there are at most $(\frac{z}{\epsilon})^{\bigO{dk}}$ many $k$-tuples of $T_1\times \cdots \times T_k$ since $|T_i| = (\frac{z}{\epsilon})^{\bigO{d}}$, yielding the claimed runtime.

For correctness, we will show that, for any group $\vecw \in \cW$, we have $\cost(\vecw,X) \le (1+\epsilon)\cost(\vecw,O)$, which implies  $\cost({\cI},X) \le (1+\epsilon)\cost({\cI},O)$, where $O =\{o_1,\cdots, o_k\}$ is an optimal solution center. Let $\Pi_i$ be the set of points in $P$ served by $o_i$, for $i \in [k]$. Let $\ell^*_i \in \Pi_i$ be a  point that is closest to $o_i$. Let this distance be $\lambda^*_i$, i.e., $\lambda^*_i := \delta(\ell^*_i,o_k)$. We call $\ell^*_i$ as the leader of $\Pi_i$ with radius $\lambda^*_i$.
Let $\lambda'_i$ be the smallest number equal to some power of $(1+\nicefrac{\epsilon}{10z})$ that is larger than $\lambda^*_i$. Then, note that $\lambda'_i \ge \lambda^*_i \ge \frac{\lambda'_i}{(1+\epsilon/10z)}$. 
Next, consider the $\frac{\epsilon}{20z}$-ball decomposition of $\ball(\ell^*_i,\lambda'_i)$, and let $b^*_i$ be a ball containing $o_i$  and let $t^*_i$ be its center. Now, consider the iteration of the algorithm corresponding to leader-radii enumeration $(\ell^*_1,\cdots,\ell^*_k)$ and $(\lambda'_1,\cdots,\lambda'_k)$ and, center enumeration $(t^*_1,\cdots, t^*_k)$. Then, for any $p \in \Pi_i, i \in [k]$, we have that
\[
\delta(p,t^*_i) \le \delta(p,o_i) + \delta(o_i,t^*_i)  \le \delta(p,o_i) + \frac{\epsilon}{10z} \lambda'_i \le  \delta(p,o_i) + \frac{\epsilon}{10z} \left(1+\frac{\epsilon}{10z}\right)\lambda^*_i \le  \left(1+\frac{\epsilon}{5z}\right) \delta(p,o_i).
\]
Hence, for any group $\vecw \in \cW$, we have
$\vecw[p] \delta(p,t^*_i)^z \le (1+\frac{\epsilon}{5z})^z \vecw[p]\delta(p,o_i)^z \le (1+\epsilon) \vecw[p]\delta(p,o_i)^z $, where we used the inequality $(1+\frac{\epsilon}{5z})^z \le e^{\epsilon/5} \le 1+ \frac{\epsilon}{5} + (\frac{\epsilon}{5})^2 + \cdots  \le 1+ \epsilon$.
Hence, $\cost(\vecw,X) = \sum_{p\in P}w[p]\delta(p,X)^z \le (1+\epsilon)\cost(\vecw,O)$.
\end{proof}
We conclude this section by proving the main claim of this section (Theorem~\ref{thm:pas_kd}) by using the results of Theorem~\ref{thm: coresetrobust} and Theorem~\ref{thm:pas_kdp} as follows.

\begin{proof}[Proof of Theorem~\ref{thm:pas_kd}]
Given an instance $\cI=(F,P,\cW)$ of \rkc, and the accuracy parameter $\epsilon>0$, we invoke Theorem~\ref{thm: coresetrobust} on $\cI$ with parameter $\epsilon/10$ to obtain an coreset $(P',\cW')$ such that  $P' \subseteq P:  |P'|= \left(\frac{2^z}{\epsilon}\right)^{\bigO{d}} kz\log n$. Let $\cI'=(F,P',\cW')$ be the resulting instance. Then, we invoke  Theorem~\ref{thm:pas_kdp} on $\cI'$ with parameter $\epsilon/10$ to obtain $X \subseteq F: |X| \le k$ such that $\cost({\cI'},X) \le (1+ \epsilon/10)\opt(\cI')$.

First we analyze the overall running time. 
With $|P'|= \left(\frac{2^z}{\epsilon}\right)^{\bigO{d}} kz\log n$, Theorem~\ref{thm:pas_kdp} runs in time  $\left(\left(\frac{2^z}{\epsilon}\right)^d kz \log n\right)^{\bigO{k}} \poly(n,m)$, leading to
$\left(\left(\frac{2^z}{\epsilon}\right)^d zk \log k\right)^{\bigO{k}} \poly(n,m)$ as the overall running time as desired. For correctness, consider
\begin{align*}
    \cost({\cI},X) &\le (1+\epsilon/10) \cost(\cI',X) \qquad \text{by the coreset property} \\
            &\le (1+\epsilon/10)^2 \opt(\cI') \qquad \text{by Algorithm~\ref{alg:epas_kd}} \\
            &\le (1+\epsilon/10)^3  \opt(\cI) \qquad \text{by the coreset property} \\
            &\le (1+\epsilon) \opt(\cI).
\end{align*}

\end{proof}
{\footnotesize
\begin{algorithm}[H] 
\caption{$(1+\epsilon)$-approximation algorithm for \rkc} \label{alg:epas_kd}
  \KwData{Instance $\cI=(F,P,\cW)$ of \rkc}
  \KwResult{$(1+\epsilon)$-approximate solution $X \subseteq F$}
    Let $X \leftarrow \emptyset$\;
    \ForAll{$k$-tuples $(\ell_1,\cdots,\ell_k)$ of $P$}{
        \ForAll{$k$-tuples $(\lambda_1,\cdots,\lambda_k)$ radii of $(\ell_1,\cdots,\ell_k)$  that are power of $(1+\nicefrac{\epsilon}{10z})$
        }{
        \For{$i \in [k]$}{
             $\cB_i \leftarrow \{\frac{\epsilon}{20z}$-ball decomposition of $\ball(\ell_i,\lambda_i)$\}\;
             } 
             $T_i \leftarrow \{f \in F \mid f \text{ is an arbitrary facility in ball } b \in \cB_i \}$
                                  \footnote{If $F=\rd$ then $T_i \leftarrow \{x_b \in F \mid x_b \text{ is the center of ball } b \in \cB_i \}$}
             \;
        \ForAll{$k$-tuples $(t_1,\cdots,t_k)$ of $T_1\times \cdots \times T_k$}{
            \If{$\cost({\cI},\{t_1,\cdots,t_k\}) < \cost(\cI,X)$}{
                $X \leftarrow \{t_1,\cdots,t_k\}$
                }
            }
        }
    }
    \Return{$X$}
\end{algorithm}
}
\begin{lemma}[Ball decomposition lemma]\label{apd:balldecomp}
Consider a metric space $(X,\delta)$
with doubling dimension $d$. A subset 
$A\subseteq X$ is {\it $\epsilon$-dense} in $X$, if $\forall x\in X$, $\exists y\in A$ such that $\delta(x,y)\leq \epsilon$.
$A$ is {\it $\epsilon$-separated}, if
$\forall x \neq y\in A$, $\delta(x,y)> \epsilon$
and  $A$ is $\epsilon$-net of $X$ if $A$ is $\epsilon$-separated as well as $\epsilon$-dense.
 Then we have the following: 
 \begin{enumerate}
\item There exists an $\epsilon$-dense set $A\subseteq \ball{(x,r)}$ of size $(\frac{r}{\epsilon})^{O(d)}$  $\forall x\in X$, $r>0$, $\epsilon \leq r/2$. 
\item For all $\epsilon$-separated set $A\subseteq \ball{(x,r)}$, $r>0$, $\epsilon \leq r/2$ it holds that $|A|\leq (\frac{r}{\epsilon})^{O(d)}$.
\end{enumerate}
\end{lemma}
\begin{proof}
\begin{enumerate}
\item Let us denote $k$ as the doubling constant of the considered metric space $(X,\delta)$ and $m=\lceil \log (r/\epsilon) \rceil$. Then by the definition, we have $k=2^{O(d)}$. Note \ball{(x,r)} can be covered with $k$ balls of radius $r/2$. Further each of these balls of radius $r/2$ can be covered with $k$ balls of radius $r/4$ resulting the original ball \ball{(x,r)} can be covered with $k^m$ balls of radius $\frac{r}{2^m}$ (note $\frac{r}{2^m}\leq \epsilon$). Again since $k=2^{O(d)}$, a simple calculation shows that $k^m=(\frac{r}{\epsilon})^{O(d)}$. Let $A$  be the centers of these balls then clearly $A$ is $\epsilon$-dense as required. 
\item  Note that the balls of radius $(\frac{\epsilon}{2})$ around the points of $A$ are disjoint as the points in the set $A$ have pairwise distance strictly greater than $\epsilon$
  and further their union is included in the $\ball{(x,r+\frac{\epsilon}{2})}$. Hence, $A$ is at most the size of any $\epsilon/2$-dense set within the $\ball{(x,r+\frac{\epsilon}{2})}$. By the previous claim, we finally prove that there exists an $\epsilon/2$-dense set in  $\ball{(x,r+\frac{\epsilon}{2})}$ of size $(\frac{2r}{\epsilon})^{O(d)}$.\\
One can even construct the set $A$ greedily. Initially $A=\phi$. Next we choose an arbitrary point in \ball{(x,r)} and add it to $A$. Let $B$ denote the union of the closed balls of radius $\epsilon$ around the points in $A$. As long as we find a point $p\in \ball{(x,r)}\setminus B$, we add this particular point to the set $A$. Note, when the algorithm stops the resulting set $A$ becomes an $\epsilon$-net as both $\epsilon$-dense as well as $\epsilon$-separated conditions are being satisfied. 
\end{enumerate}
\end{proof}

\section*{Acknowledgements}
Fateme Abbasi and Jaros\l{}aw Byrka were supported by the Polish National Science Centre (NCN) Grant 2020/39/B/ST6/01641. 
Sandip Banerjee acknowledges the support by SNSF Grant 200021 200731/1 and also the support of Polish National Science Centre (NCN) Grant 2020/39/B/ST6/01641 while at the University of Wroc\l{}aw, Poland.
Parinya Chalermsook, Kamyar Khodamoradi, and Joachim Spoerhase were supported by the European Research Council (ERC) under the European Union’s Horizon 2020 research and innovation programme (grant agreement No 759557).  
Ameet Gadekar was supported by the European Research Council (ERC) under the European Union’s Horizon 2020 research and innovation programme (grant agreement No 759557) while at Aalto University, and by the Israel Science Foundation (grant No. 1042/22).

 \bibliographystyle{plain}
 \bibliography{references}

\end{document}